\documentclass[journal,10pt]{IEEEtran}
\usepackage[utf8]{inputenc}

\usepackage{amsmath}
\usepackage{amssymb}
\usepackage{amsthm}
\usepackage{amsfonts}
\usepackage{cite}

\usepackage{tabularx}
\usepackage{makecell}
\usepackage{array}

\usepackage[caption=false]{subfig}
\usepackage{booktabs}
\usepackage{bbm}
\usepackage{mathrsfs}

\DeclareMathOperator*{\argmin}{arg\,min}

\newcommand{\Rmnum}[1]{\expandafter\@slowromancap\romannumeral #1@}
\usepackage{graphicx}

\usepackage{epsfig}
\usepackage[numbers,sort&compress]{natbib}
\usepackage{xcolor}
\usepackage{color}
\usepackage{enumerate}
\usepackage{extarrows}
\usepackage[ruled,linesnumbered,lined]{algorithm2e}  
\usepackage{algorithmic}
\usepackage{tensor}
\usepackage[colorlinks,linkcolor=blue,anchorcolor=blue,citecolor=blue,urlcolor=black]{hyperref}
\usepackage{url}

\begin{document}

\title{From High-Level Communication  Requirements to KPIs via Conformal Signal Temporal Logic Learning}
\author{ Jiechen Chen, \IEEEmembership{Member,~IEEE}, Noemi Giustini, \IEEEmembership{Member,~IEEE}, Michele Polese,  \IEEEmembership{Member,~IEEE}, \\and Osvaldo Simeone,~\IEEEmembership{Fellow,~IEEE}
\thanks{J. Chen is with the Department of Engineering, King’s College London, London, WC2R 2LS, UK (email: jiechen.chen@kcl.ac.uk). N. Giustini and M. Polese are with the Institute for Intelligent Networked Systems, Northeastern University, Boston, MA 02115 USA (email:
giustini.n@northeastern.edu; m.polese@northeastern.edu). O. Simeone is with the Institute for Intelligent Networked Systems,
Northeastern University London, One Portsoken Street, London, E1 8PH,
UK (email: o.simeone@nulondon.ac.uk). \\
This work was supported by the European Research Council (ERC) under the European Union’s Horizon Europe Programme (grant agreement No. 101198347), by an Open Fellowship of the EPSRC (EP/W024101/1), by the EPSRC project (EP/X011852/1), and by the U.S. NSF under grant TI-2449452.
 }
\vspace*{-0.8cm}
}

\maketitle

\IEEEpeerreviewmaketitle
\newtheorem{definition}{\underline{Definition}}[section]
\newtheorem{fact}{Fact}
\newtheorem{assumption}{Assumption}
\newtheorem{theorem}{Theorem}
\newtheorem{lemma}{\underline{Lemma}}[section]
\newtheorem{proposition}{\underline{Proposition}}[section]
\newtheorem{corollary}[proposition]{\underline{Corollary}}
\newtheorem{example}{\underline{Example}}[section]
\newtheorem{remark}{\underline{Remark}}[section]
\newcommand{\mv}[1]{\mbox{\boldmath{$ #1 $}}}
\newcommand{\mb}[1]{\mathbb{#1}}
\newcommand{\Myfrac}[2]{\ensuremath{#1\mathord{\left/\right.\kern-\nulldelimiterspace}#2}}
\newcommand\Perms[2]{\tensor[^{#2}]P{_{#1}}}
\newcommand{\note}[1]{[\textcolor{red}{\textit{#1}}]}

\begin{abstract}
Softwarized radio access networks (RANs), such as those based on the Open RAN (O-RAN) architecture, generate rich streams of key performance indicators (KPIs) that can be leveraged to extract actionable intelligence for network optimization. However, bridging the gap between low-level KPI measurements and high-level requirements, such as quality of experience (QoE), requires methods that are both {relevant}, capturing temporal patterns predictive of user-level outcomes, and {interpretable}, providing human-readable insights that operators can validate and act upon. This paper introduces {conformal signal temporal logic learning} (C-STLL), a framework that addresses both requirements. C-STLL leverages signal temporal logic (STL), a formal language for specifying temporal properties of time series, to learn from data interpretable formulas that characterize key properties of KPI traces satisfying high-level requirements. We first formulate the problem of reliably identifying an accurate STL formula as a risk-averse inference problem. We show that the optimal solution requires the construction of a set of STL formulas with formal coverage guarantees. Based on this theoretical insight, the proposed C-STLL extracts from calibration data a set of STL formulas with such guarantees via Learn Then Test, while also filtering for  formula complexity and diversity. Experiments on a mobile gaming scenario in ns-3 and on KPI traces measured on a softwarized RAN testbed demonstrate that C-STLL effectively controls risk below target levels, while returning interpretable temporal specifications that relate KPI behavior to QoE outcomes.
\end{abstract}


\section{Introduction}\label{sec:intro}

\subsection{Motivation}

The wealth of data produced by modern wireless networks creates a unique opportunity to extract actionable intelligence that can drive network optimization, anticipate failures, and ensure quality of service guarantees \cite{brik2022deep}. In particular, network controllers in disaggregated architectures like O-RAN can aggregate rich streams of key performance indicators (KPIs), at the radio access network (RAN), including throughput, latency, packet loss, and signal quality \cite{polese2023understanding, bonati2021intelligence, lacava2023programmable}. Beyond relevance, intelligence extracted from the RAN must also be \emph{interpretable} \cite{guo2020explainable, brik2024explainable,fiandrino2023explora}, allowing network operators to understand \emph{why} a decision was made either to validate decisions against domain expertise or to satisfy regulatory requirements. 

However, extracting interpretable intelligence from RAN traces presents a challenge given the gap between the \emph{low-level KPI measurements} readily available at the RAN and the \emph{high-level quality-of-experience (QoE) requirements} that matter to network operators and end users (see Fig. \ref{gaming}). To be effective, explanations about QoE outcomes must be constructed only from the data available in the part of the network where the control is taken. For RAN optimization, this means low-level RAN KPI traces. This necessitates the development of novel interpretable mechanisms to bridge the semantic gap between QoE targets and RAN KPI traces.

\begin{figure}[t]
    \centering
\includegraphics[width=1.0\linewidth]{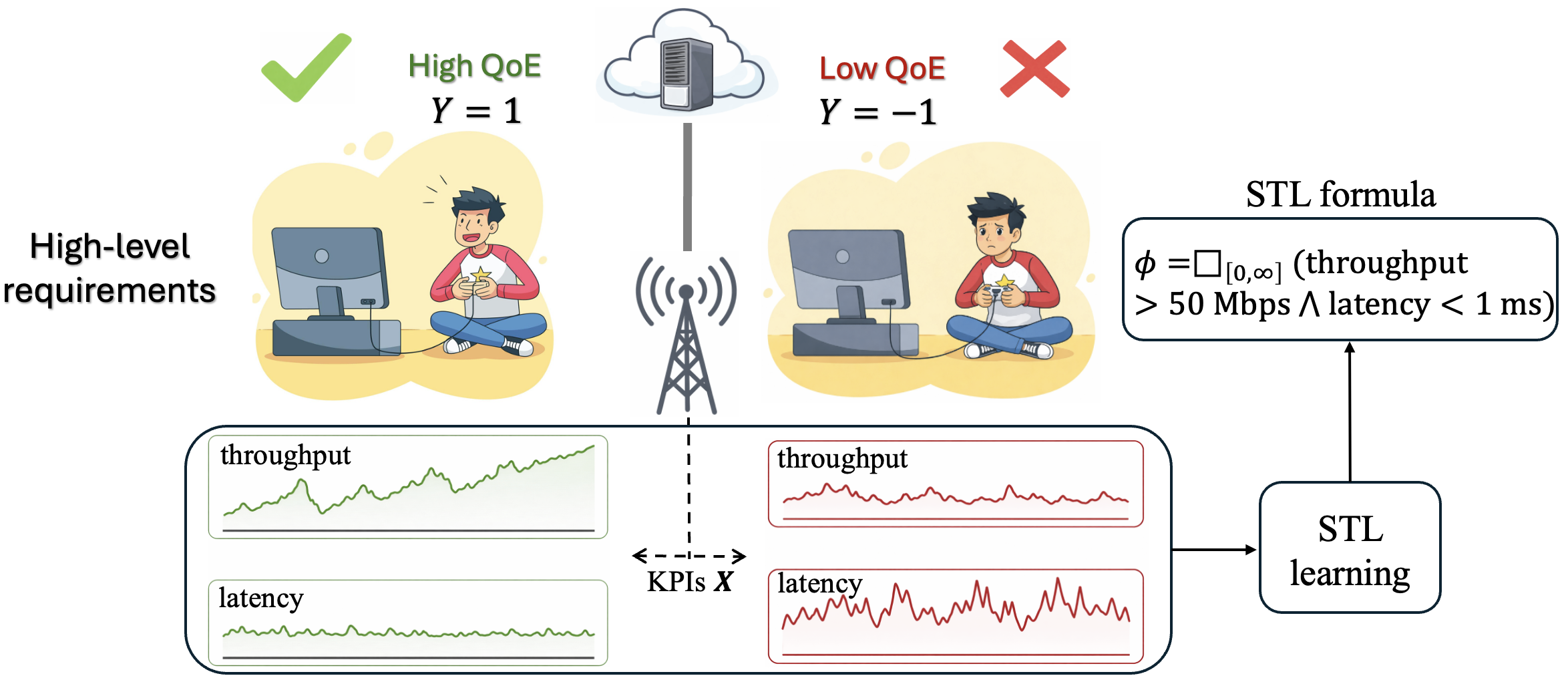}
    \caption{From high-level requirements to KPI requirements: The radio access network (RAN) collects traces of KPIs such as throughput and latency over time. High-level QoE evaluations are provided by end users. The goal of this work is to infer interpretable properties of KPI traces that are predictive of whether high-level requirements are satisfied or not. These properties are expressed by the formal language signal temporal logic (STL), which provides tools to describe temporal constraints as time series.  }
    \label{gaming}
\end{figure}

\emph{Signal Temporal Logic} (STL) offers a principled framework for expressing and learning interpretable temporal properties of time-series data \cite{maler2004monitoring, donze2010robust}. An STL formula such as \begin{equation}\label{eq:firstSTL}\Box_{[0,\infty)}((\text{latency} < 100~\text{ms}) \land \Diamond_{[0,5]}(\text{throughput} > 50~\text{Mbps}))\end{equation} provides a human-readable specification that directly relates KPI behavior to requirements. For example, the STL formula in (\ref{eq:firstSTL}) states that it must always ($\Box_{[0,\infty)}$) hold, at every TTI $t$, that the latency remains below 100 ms and ($\land$) the throughput exceeds 50 Mbps at least once ($\Diamond_{[0,5]}$) within the following 6-TTI window $[t,t+5]$. Unlike feature attributions from standard explainable AI methods such as SHAP \cite{lundberg2017unified} or LIME \cite{ribeiro2016should}, STL formulas constitute \emph{formal specifications} that can be verified, composed, and directly translated into network policies.

That said, existing standard \emph{STL learning} (STLL) methods, while effective at extracting candidate formulas from labeled data, provide no formal guarantees on the reliability of learned specifications \cite{li2024tlinet, bombara2021offline}. A formula that performs well on training data may fail to generalize, potentially leading to costly misclassifications in deployment. This limitation is particularly concerning in wireless communications, where decisions on network configuration based on unreliable predictions can degrade user experience or violate service level agreements (SLAs).

\begin{figure*}[htp]
    \centering
\includegraphics[width=0.7\linewidth]{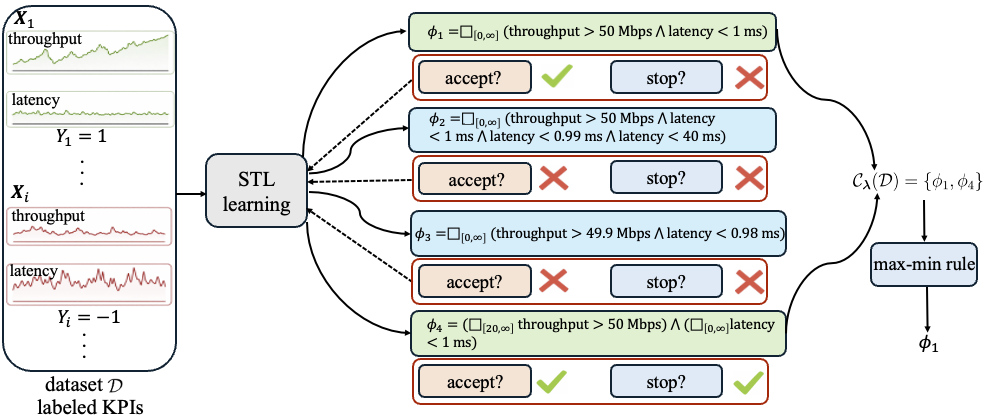}
    \caption{Illustration of the proposed conformal STL learning (C-STLL) scheme. As illustrated in Fig. \ref{gaming},  KPI traces are labeled as $Y=1$ or $Y=-1$ based on whether or not they correspond to settings in which positive or negative higher-level requirements are satisfied. Via any STLL algorithm \cite{li2024tlinet}, the dataset $\mathcal{D}$ of labeled KPI traces can be leveraged to produce an STL formula $\phi$,   providing interpretable conditions on the KPI time series that are consistent with positive labels. However, this approach provides no guarantees on the reliability  and accuracy of the formula $\phi$. The proposed C-STLL wraps around any STLL procedure, such as \cite{li2024tlinet}, to produce a \emph{set} of STL formulas with \emph{formal reliability guarantees}. Through a sequential application of the STL algorithm, each newly learned formula is either accepted for inclusion in the set or rejected, and the process is continued or stopped, based on calibrated threshold-based procedures. In this example, formula  $\phi_2$ is rejected due to its high complexity, while  $\phi_3$ is rejected for being  too similar to $\phi_1$, which is already in the set. The formula deployed at test time is selected from $\mathcal{C}(D)$ using the max-min agreement rule in \eqref{action}.}
    \label{model}
\end{figure*}

\subsection{Main Contributions}

This paper introduces \emph{conformal STL learning} (C-STLL), a novel framework for extracting interpretable, reliability-guaranteed temporal specifications from wireless KPI traces. The main contributions are as follows:

\noindent $\bullet$ \textbf{ Interpretable KPI-to-QoE specification learning:} We demonstrate that STL provides a natural and effective formalism for expressing temporal properties of KPI traces that are predictive of high-level QoE requirements. Unlike black-box ML models, learned STL formulas are human-readable and verifiable.
    
\noindent $\bullet$ \textbf{Risk-averse specification learning:} We cast the learning of a deployable STL specification as a risk-averse decision problem, and show that its solution reduces to a set-valued predictor combined with a max-min decision rule. Based on this theoretical insight, we then develop C-STLL, a calibration methodology that wraps around STLL algorithms to produce a \emph{set} of STL formulas with formal reliability guarantees (see Fig. \ref{model}). Specifically, C-STLL ensures that, with high probability, the returned set contains at least one formula achieving a target accuracy level. This is accomplished through a novel combination of sequential STL formula generation, acceptance rules based on complexity and diversity criteria, and multiple hypothesis testing via the Learn Then Test (LTT) framework \cite{angelopoulos2025learn,laufer2022efficiently,quach2023conformal}.

\noindent $\bullet$ \textbf{Experimental validation:} We evaluate C-STLL using the ns-3 network simulator, modeling a mobile gaming scenario with realistic traffic patterns and network conditions, and using KPI traces measured on a softwarized RAN testbed \cite{Villa_2025}. Experiments demonstrate that C-STLL produces interpretable temporal specifications while achieving favorable classification accuracy--complexity trade-offs.

\subsection{Further Related Work}

\subsubsection{Signal temporal logic}

STL learning (STLL) algorithms can extract STL rules from labeled traces using different methods such as gradient descent for a neural model \cite{li2024tlinet}, decision-tree induction \cite{bombara2021offline, bombara2016decision}, and evolutionary search \cite{nenzi2018robust, brunello2023monitors}. An alternative class of STL synthesis techniques translate requirements expressed in natural language directly into STL
specifications \cite{he2022deepstl, fang2025enhancing, fushimi2026structural}, thus not requiring any data. As discussed in Sec. \ref{seq}, the proposed method is agnostic to the choice of STLL algorithm.

While STLL methods have been successfully applied in cyber-physical systems \cite{bartocci2018specification}, its application to wireless communications remains largely unexplored. Reference \cite{shui2025analysis} learns STL formulas characterizing normal vehicular trajectories to detect spoofing attacks on integrated sensing and communication systems, though without reliability guarantees.

\subsubsection{Explainable AI for wireless communications}

The need for explainable AI (XAI) in wireless networks has been increasingly recognized as ML models are deployed for critical network functions \cite{guo2020explainable,brik2024explainable}. However, existing XAI methods for wireless predominantly rely on statistical attribution, and the application of formal specification languages to wireless network intelligence remains an open research direction.

\subsubsection{Conformal prediction and reliability guarantees}

Conformal prediction  is a distribution-free framework for constructing prediction sets with guaranteed coverage \cite{vovk2005algorithmic, shafer2008tutorial}, with applications including  \cite{cohen2023calibrating,simeone2025conformal, 10682971}. Conformal prediction was leveraged in \cite{li2025conformal, wang2026conformalized} to calibrate the label decision induced by a single learned STL formula. In terms of applications, reference \cite{cairoli2025conformal} studied the problem of predictive monitoring based on STL and conformal prediction.

The Learn Then Test (LTT) framework recasts risk control as multiple hypothesis testing, enabling control of general risk functions beyond miscoverage \cite{angelopoulos2025learn, farzaneh2026statistically}. LTT was leveraged to calibrate several communication protocols \cite{simeone2025conformal,10682971}. LTT has also been applied to the problem of generating STL specifications from natural language \cite{wang2026scp}. 

\subsubsection{QoE prediction}
Predicting user QoE from network-level measurements \cite{barakabitze2019qoe} was addressed in reference \cite{charonyktakis2016user} for voice-over-IP (VoIP) systems, while reference \cite{kougioumtzidis2023deep} focused on predicting future virtual reality QoE levels in O-RAN networks. These approaches inherit the limitations of black-box models in terms of interpretability and statistical validity.

\subsection{Organization}
The remainder of this paper is organized as follows. Section~\ref{sec:stl} reviews the syntax and semantics of STL. Section~\ref{sec:stll} formulates the STL learning problem and describes existing solution approaches. Section~\ref{sec:cstll} presents the proposed C-STLL framework, including the sequential generation procedure and the calibration methodology based on multiple hypothesis testing. Section~\ref{sec:experiments} provides experimental results, and Section~\ref{sec:conclusion} concludes the paper.

\section{Background: Signal Temporal Logic } \label{sec:stl}
Signal Temporal Logic (STL) is a formal language for specifying and analyzing the temporal properties of time-series data. While STL is defined over continuous time, here we adopt the standard discrete-time restriction \cite{li2024tlinet, lindemann2021stl, rodionova2022temporal} to account for the fact that KPI signals are not available continuously, but are reported periodically, with a reporting period spanning one or more TTIs \cite{polese2023understanding}. In this section, we present the standard STL syntax together with quantitative metrics used to evaluate the degree of satisfaction of given STL properties. 

\subsection{STL Formulas}

Let $\mv X=[\mv x_0, \mv x_1, \ldots, \mv x_{T-1}]$ be a finite discrete-time trace with $\mv x_t =[x_{1,t},...,x_{d,t}] \in \mathbbm{R}^d$ denoting the system state at time $t$, where integer $d$ is the dimension of the state. In wireless communications, each vector $\mv x_t$ may represent a collection of KPIs extracted over transmission time intervals (TTIs) indexed as $t=0,1,\ldots$, e.g., one entry $x_{i,t}$ of vector $\mv x_t$ may report the throughput of some user and another entry, $x_{j,t}$,  the latency of the user. 

An STL formula $\phi$ specifies a property that time series $\mv X$ may exhibit or not. Specifically, one writes \begin{align}
    (\mv X, t) \models \phi,
\end{align}
if trace $\mv X$ satisfies property $\phi$ from time $t$ onward, so the symbol $\models$ reads ``satisfies''. For example, a trace  of KPIs, $\mv X$, corresponding to nominal behavior of a radio access network may satisfy the property $\phi$ that the latency of a certain group of users does not exceed 1 ms for more than a given number of TTIs starting from some TTI $t$. 

STL properties $\phi$ are built from atomic predicates of the form 
\begin{align}
    \mu ::= \mv a^\top \mv x > b, \label{mu}
\end{align}
with $\mv a\in \mathbbm{R}^d$ and $b\in \mathbbm{R}$, where the symbol $::=$ is to be read as ``is defined as". Note that $\mv a^\top \mv x$ is a scalar, so that (\ref{mu}) is a single inequality
defining a halfspace, and not a set of coordinate-wise comparisons. Specifically, the basic STL property  $(\mv X,t) \models \mu$ holds if and only if $\mv a^\top \mv x_t > b$. As we will discuss below, more complex properties can be built from basic predicates of the form \eqref{mu}.

For example, consider the property that the throughput $x_{i,t}$ of a given user must be larger  than a threshold $b= 50 $ Mbit per second (Mbps). Using a ``one-hot'' vector $\mv a$ with a single one in position $i$ and all other zero entries, denoted as $\mv e_i$, this can be expressed as the atomic predicate $\mu ::=\mv e_i ^\top \mv x > 50  \text{ Mbps}$. 

More generally,  STL formulas can be constructed by leveraging the \emph{eventually} operator  $\Diamond_{[t_1, t_2]}$, the \emph{always} operator $\Box_{[t_1, t_2]}$, as well as logical operations such as AND and OR. In particular, using the eventually operator, the condition $(\mv X, t) \models \phi$ with the STL property 
\begin{align}
    \phi ::= \Diamond_{[t_1, t_2]} \mu \label{dia}
\end{align}
is true if the condition $\mu$ holds at \emph{any} time $t^{\prime}\in[t+t_1,t+t_2]$, i.e., in an interval of the form $[t_1,t_2]$ with time measured starting from index $t$. In a similar way, using the always operator, the condition $(\mv X, t) \models \phi$ with the STL property 
\begin{align}
    \phi ::= \Box_{[t_1, t_2]} \mu \label{box}
\end{align}
is valid if the property $\mu$ holds for \emph{all} $t^{\prime}\in[t+t_1,t+t_2]$. For instance, continuing the wireless example, the STL property $(\mv X, 0) \models \Box_{[0, \infty)} ( x_i> 50
\text{Mbps})$ is true if the throughput $x_{i,t'}$ is no smaller than $50
\text{Mbps}$ for all TTIs starting at time index $0$. 

In broadest generality, formulas such as \eqref{dia} and \eqref{box} can be composed recursively, making use also of logical operations, namely AND ($\land$) and OR ($\lor$). Specifically, one can recursively define an STL property via the rule
\begin{align}
    \phi ::= \mu| \land_{i=1}^n \phi_i |\lor_{i=1}^n \phi_i| \Diamond_{[t_1, t_2]} \phi | \Box_{[t_1, t_2]} \phi, \label{stl}
\end{align}
where $\mu$ is an atomic predicate (\ref{mu}), $n$ is a positive integer, and each step of the recursion applies one of the options in \eqref{stl}, which are separated by a bar $``|"$ following a standard convention  \cite{donze2013signal}.

For instance, consider a requirement that combines a per-TTI throughput constraint with a periodic latency constraint for a given user. To model this requirement, the STL property
\begin{align}
    (\mv X, 0) \models \Box_{[0,\infty)}\big(x_{i,t}> 50 ~\text{Mbps} \land \Diamond_{[0, 5]}(x_{j,t}\leq  1 \text{ ms}) \big)
\end{align}
stipulates that, for each time step $t$ starting from $t=0$, the throughput is larger than 50 Mbps and the latency is  no larger than 1 ms in at least once within the following six time steps.

\subsection{Robustness}

STL is equipped with a robustness score, which assigns to each formula $\phi$ and signal $\mv X$ a real-valued function $\rho(\mv X, \phi, t) \in\mathbbm{R}$ that measures the degree to which $\phi$ is satisfied when the evaluation begins at time $t$ \cite{donze2010robust, fainekos2009robustness}. By definition of robustness measure, a signal  satisfies the formula $\phi$ from time $t$ if and only if its robustness is positive, i.e.,
\begin{align}
    (\mv X, t)\models \phi \Longleftrightarrow
 \rho(\mv X, \phi, t)>0,
\end{align}
and larger robustness values $\rho(\mv X, \phi, t)$ indicate stronger satisfaction.

For an atomic predicate $\mu$ in \eqref{mu}, the robustness is defined as
\begin{equation}
\rho(\mv X, \mu, t) = \mv a^\top \mv x_t - b. \label{affine}
\end{equation} Accordingly, the robustness is positive if $\mv a^\top \mv x_t >b$.  When considering a general STL formula constructed via the rule (\ref{stl}), the robustness can be obtained by using the following rules. The robustness of the logical AND of more properties,  $\phi=\land_{i=1}^n \phi_i$, is the minimum
\begin{align}
\rho(\mv X, \land_{i=1}^n \phi_i, t)= \min_{i=1,\ldots,n} \rho(\mv X,\phi_i, t),  \label{s1}
\end{align} since all properties $\{\phi_i\}_{i=1}^n$ must be satisfied; while for the logical-OR formula $\phi=\lor_{i=1}^n \phi_i$, the robustness is the maximum
\begin{align}
\rho(\mv X, \lor_{i=1}^n \phi_i, t)= \max_{i=1,\ldots,n} \rho(\mv X,\phi_i, t), \label{s2}
\end{align} since the validity of any property $\phi_i$ is sufficient for the validity of the entire property $\phi$. Finally, in a similar way, for the eventually operator, we have \begin{align}
\rho(\mv X, \Diamond_{[t_1,t_2]} \phi, t)= \max_{t' \in [t+t_1,t+t_2]} \rho(\mv X, \phi, t'), \label{te2}
\end{align} and for the always operator the robustness is
\begin{align} 
\rho(\mv X, \Box_{[t_1,t_2]} \phi, t)= \min_{t' \in [t+t_1,t+t_2]} \rho(\mv X, \phi, t'). \label{te1}
\end{align}

\section{Learning STL Properties} \label{sec:stll}

In this section, we discuss the problem of learning from labeled traces an interpretable prediction rule based on STL by following \cite{li2024tlinet}. The goal is to infer a formal property, in the form of an STL formula $\phi$, that lower-level traces $\mv X$ satisfying higher-level requirements are likely to meet.

\subsection{Formulating STL Learning} \label{stll}
Consider a system in which traces $\mv X$ can be  collected that describe the evolution of the variables of interest, here the KPIs of a RAN. To each trace $\mv X$, an end user can potentially assign a binary label  $Y\in\{1,-1\}$ indicating whether trace $\mv X$ satisfies  given higher-level requirements, $Y=1$, or violates them, $Y=-1$. Henceforth, we will refer to a trace $\mv X$ labeled $Y=1$ as \emph{positive}, and to  a trace $\mv X$ with $Y=-1$ as \emph{negative}. As in Fig. \ref{gaming}, the higher-level requirements may refer to the QoE of end users when interacting with certain applications over the network.

We are interested in the problem of \emph{STL learning} (STLL), where the goal is to infer an STL formula ${\phi}$ that can effectively distinguish between traces with positive and negative labels from data. Specifically,  we would like to identify an STL formula $\phi$ that is likely to be satisfied by traces $\mv X$ with positive labels, i.e., $(\mv X,0) \models \phi$ if $Y=1$, while being violated by traces with negative labels, i.e., $(\mv X,0) \not\models \phi$ if $Y=-1$. Note that the starting time is chosen as $t=0$ by convention. 

Accordingly, the classification decision $\hat{Y}$ associated with an STL property $\phi$ is given by \begin{equation}
\hat{Y}= c_{\phi}(\mv X)= \begin{cases}
1, & \text{ if }(\mathbf{X},0)\models\phi,\\
-1, & \text{ if }(\mathbf{X},0)\not\models\phi, \label{pred0}
\end{cases}
\end{equation} or equivalently in terms of robustness \begin{equation}
c_{\phi}(\mv X)=\begin{cases}
1, & \text{ if } \rho(\mv X, \phi, 0)>0,\\
-1, & \text{ otherwise}. \label{pred}
\end{cases}
\end{equation} Accordingly, the robustness $\rho(\mv X, \phi, 0)$ serves as a logit in conventional binary classification. 

The STLL problem is formulated as a supervised learning task. To this end, assume access to a labeled dataset $\mathcal{D}=\{(\mv X_i, Y_i)\}_{i=1}^{|\mathcal{D}|}$ in which each pair $(\mv X_i, Y_i)$ is drawn i.i.d. from an underlying ground-truth distribution $P_{\scalebox{0.7}{\mv X}, Y}$ following the standard frequentistic learning framework (see e.g.,  \cite{simeone2022machine}). The data distribution describes the current, ground-truth, relationship between the traces $\mv X$ and the label $Y$. 

Using \eqref{pred}, for each labeled pair $(\mv X, Y)\sim P_{\scalebox{0.7}{\mv X}, Y}$, the \emph{classification margin}, given by the product $Y\rho(\mv X,\phi,0)$, is positive if the STL formula correctly marks a trace $\mv X$ as positive or negative; while a negative classification margin $Y\rho(\mv X,\phi,0)$ indicates an incorrect classification decision. Accordingly, any non-increasing function of the classification margin can serve as a loss function for the STL problem (see, e.g., \cite[Chapter 6]{simeone2022machine}). Specifically, in \cite{li2024tlinet}, the shifted hinge loss 
\begin{align}
    \ell(\phi, \mv X, Y) = \text{ReLU}(\beta -Y \rho(\mv X, \phi, 0)) - \gamma \beta \label{loss}
\end{align}
was adopted, where the hyperparameter $\beta>0$ specifies the desired robustness margin and $\gamma>0$ adjusts the trade-off in penalizing samples that already satisfy the margin.

Using the loss $\ell(\phi, \mv X, Y)$, for any given point distribution $P_{\scalebox{0.7}{\mv X}, Y}$, the goal is to find an STL rule $\phi$ that minimizes the population loss
\begin{align}
    L_p(\phi)=\mathbbm{E}_{(\scalebox{0.7}{\mv X}, Y) \sim P_{\scalebox{0.6}{\mv X}, Y}}[\ell(\phi, \mv X, Y)], \label{ploss}
\end{align}
where the average is evaluated over the joint distribution $P_{\scalebox{0.7}{\mv X}, Y}$ of input $\mv X$ and target $Y$. The expected value in \eqref{ploss} is practically approximated using the  dataset $\mathcal{D}$, yielding the training loss
\begin{align}
    L_{\mathcal{D}}(\phi) = \frac{1}{|\mathcal{D}|}\sum_{i=1}^{|\mathcal{D}|} \ell(\phi, \mv X_i, Y_i).
    \label{tloss}
\end{align}
However, a direct optimization of the training loss $L_{\mathcal{D}}(\phi)$ over the STL formula $\phi$ is made difficult by the fact that the space of STL formulas is discrete.

\subsection{Addressing the STL Learning Problem} \label{STLL}
In order to address the STLL problem of minimizing the training loss \eqref{tloss} over the STL formula $\phi$, reference \cite{li2024tlinet} relaxed the discrete variables dictating the trace of temporal operators and logical operations applied in the construction rule \eqref{stl}. The approach assumes a template for the STL formula $\phi$ as a trace of steps in the construction \eqref{stl} of an STL rule. An example, adapted from \cite{li2024tlinet}, is shown in Fig. \ref{TLI}, in which the formula $\phi$ is constructed via the following steps: 1) a number of predicates $\mu$, as in  \eqref{mu}, are selected; 2) the predicates are combined via Boolean operators, which may be either $\land$ or $\lor$; 3) temporal operators, namely $\Diamond_{[t_1,t_2]}$ or $\Box_{[t_1,t_2]}$, are applied to the resulting partial STL formulas; and 4) another Boolean operator, either $\land$ or $\lor$, is used to combine the partial STL properties constructed at the previous steps.

\begin{figure}[t!]
    \centering
    \includegraphics[width=0.6\linewidth]{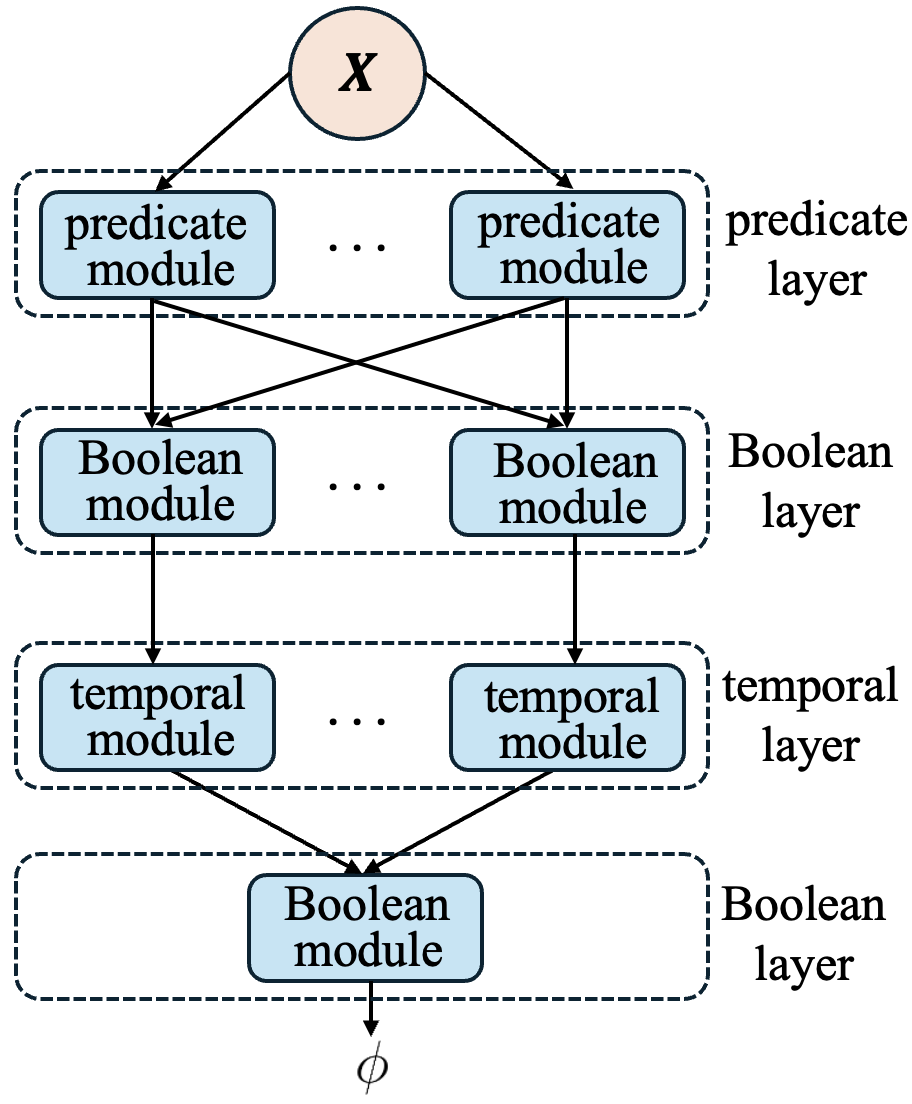}
    \caption{An illustration of a parameterized template for the construction of an STL formula $\phi$. In this example, the STL formula $\phi$ is constructed by the following steps: 1) a number of predicates $\mu$, as in  \eqref{mu}, are selected; 2) the predicates are combined via Boolean operators, either $\land$ or $\lor$; 3) temporal operators $\Diamond_{[t_1,t_2]}$ or $\Box_{[t_1,t_2]}$ are applied to the resulting partial STL formulas; and 4) another Boolean operator is used to combine the partial STL properties constructed at the previous steps.  Each layer is specified by parameters that are optimized by addressing the STLL problem. }
    \label{TLI}
\end{figure}

Accordingly, as illustrated in Fig. \ref{TLI}, a general template contains a predicate layer, Boolean layers, and temporal layers. Therefore, templates such as that in Fig. \ref{TLI} fully specify an STL property $\phi$ once a number of parameters are given, namely: \emph{(i)} the variables $\mv a$ and $b$ defining the atomic predicate of a predicate module; \emph{(ii)} the binary variables determining which logical operator (AND or OR) is applied and which formulas constructed from previous steps are combined in a Boolean module; \emph{(iii)} the binary variables determining whether a temporal operator is “always’’ $\Box$ or “eventually’’ $\Diamond$, as well as the corresponding interval bounds $t_1$ and $t_2$, in a temporal module.

Minimizing the training loss $L_{\mathcal{D}}(\phi)$ over an STL formula $\phi$ following a given template thus requires optimizing over both continuous variables -- namely $\mv a, b, t_1, t_2$ -- and discrete variables -- i.e., the mentioned binary indicators. Reference \cite{li2024tlinet} proposes relaxing the binary variables to probabilities so as to obtain a differentiable training loss, while adding regularizers to push optimized probabilities towards either 0 or 1.  
For completeness, more details on the STLL procedure introduced in \cite{li2024tlinet} can be found in Appendix A.

\section{Conformal STL Learning} \label{sec:cstll}
As discussed in the previous section, given a dataset $\mathcal{D}$, STLL methods obtain an STL formula $\phi$ that ideally distinguishes positive and negative KPI traces with minimal average error. However, being based on standard supervised learning methodologies, such standard techniques do not provide any formal, finite-sample, guarantee on the actual average loss obtained by the optimized formula $\phi$. This can make the deployment of these methods in sensitive domains potentially problematic. In this section, we introduce \emph{conformal STLL} (C-STLL), a novel STLL paradigm that provides formal finite-sample reliability guarantees building on risk-averse decision making \cite{kiyani2025rac, simeone2026decision}, test-time scaling \cite{muennighoff2025s1}, and multiple hypothesis testing \cite{angelopoulos2025learn, farzaneh2026statistically}.

\subsection{Risk-Averse STL-Based Classification}
An STL rule $\hat{\phi}$ is ideally evaluated using the population-level probability of correct detection $\Pr[c_{\hat{\phi}}(\mv X)=Y]$ for the STL classification \eqref{pred0} (or \eqref{pred}). This probability is practically estimated on a held-out validation set $\mathcal{D}^{\rm v}=\{(\mv X_i, Y_i)\}_{i=1}^{|\mathcal{D}^{\rm v}|}$ with samples generated i.i.d. from the data distribution $P_{\scalebox{0.7}{\mv X}, Y}$ as
\begin{align}
    A_{\mathcal{D}^{\rm v}}(\hat\phi)=\frac{1}{|\mathcal{D}^{\rm  v}|}
    \sum_{i=1}^{|\mathcal{D}^{\rm v}|}
    \mathbbm{1}\!\left(c_{\hat{\phi}}(\mv X_i)=Y_i\right). \label{accu}
\end{align}
While the accuracy measure \eqref{accu} requires knowledge of the true labels $\{\mv Y_i\}_{i=1}^{|\mathcal{D}^{\rm v}|}$, in practice one may only have access to the corresponding validation inputs $\mathcal{X}^{\rm v}=\{\mv X_i\}_{i=1}^{|\mathcal{D}^{\rm v}|}$. Based on set $\mathcal{X}^{\rm v}$, one can validate the agreement between a candidate rule $\hat{\phi}$ and a reference rule $\phi$ on the validation set  as
\begin{align}
    A_{\mathcal{X}^{\rm v}}(\hat\phi,\phi)=\frac{1}{|\mathcal{D}^{\rm v}|}
    \sum_{i=1}^{|\mathcal{D}^{\rm v}|}
    \mathbbm{1}\!\left(c_{\hat{\phi}}(\mv X_i)=c_{\phi}(\mv X_i)\right). \label{agree}
\end{align}

In this subsection, we address the problem of maximizing the agreement $A_{\mathcal{X}^{\rm v}}(\hat\phi(\mathcal{D}),\phi^{*})$ between the STL rule $\hat\phi(\mathcal{D})$ learned from the dataset $\mathcal{D}$ and the population-optimal rule $\phi^*=\argmin_{\phi} L_p(\phi)$ with population loss in \eqref{ploss}. As explained in Section III, the STL rule $\phi^*$ represents the ideal output of STLL. Hence, controlling the agreement $A_{\mathcal{X}^{\rm v}}(\hat\phi(\mathcal{D}),\phi^{*})$ allows us to constrain the level of suboptimality of the rule $\hat{\phi}(\mathcal{D})$ in terms of the decision made by the rule on the validation set $\mathcal{D}^{\rm v}$. We specifically adopt a risk-averse decision framework \cite{kiyani2025rac} given the sensibility of QoE-related decisions in networks. The theoretical insights developed here will be then leveraged in the next subsection to design a novel STLL calibration framework.

In the considered risk-averse formulation of the problem of  maximizing the agreement \eqref{agree} for the learned STL rule $\hat\phi(\mathcal{D})$, we target the goal of ensuring a minimum agreement level $A_{\mathcal{D}, \mathcal{D}^{\rm v}}^{\rm min}$, i.e.,
\begin{align}
    A_{\mathcal{X}^{\rm v}}(\hat\phi(\mathcal{D}),\phi^{*}) \geq A_{\mathcal{D}, \mathcal{D}^{\rm v}}^{\rm min} \label{cons}
\end{align}
with probability at least $1-\epsilon$ for some $\epsilon>0$. To formalize this notion,  we introduce an indicator $\tau$ to identify a task, i,e, a specific set of network and traffic conditions and service requirements. We model the task $\tau \sim P_{\tau}$ as a random variable with an unknown distribution $P_{\tau}$, while the training dataset $\mathcal{D}=\{(\mv X_i, Y_i)\}_{i=1}^{|\mathcal{D}|}$ and the validation dataset $\mathcal{D}^{\rm v}=\{(\mv X_i, Y_i)\}_{i=1}^{|\mathcal{D}^{\rm v}|}$ contain i.i.d. data from an unknown task-specific distribution $P_{\scalebox{0.7}{\mv X}, Y| \tau}$. The resulting mixture distribution of dataset pair $(\mathcal{D}, \mathcal{D}^{\rm v})$ is denoted as $P_{\mathcal{D}\mathcal{D}^{\rm v}}$. All averages and probabilities over dataset pairs in the rest of this subsection are understood to be with respect to this distribution. Within this framework, we address the risk-averse STLL problem of maximizing the average generated agreement $\mathbb{E}[A^{\rm min}_{\mathcal{D}, \mathcal{D}^{\rm v}}]$ subject to the minimum-agreement constraint \eqref{cons}, i.e.,
\begin{align}
    \max_{\hat\phi(\cdot),\,A^{\rm min}_{\mathcal{D},\mathcal{D}^{\rm v}}}\;
    &\mathbb{E}_{(\mathcal{D},\mathcal{D}^{\rm v})\sim P_{\mathcal{D}\mathcal{D}^{\rm v}}}\!
    \left[A^{\rm min}_{\mathcal{D},\mathcal{D}^{\rm v}}\right] \notag \\
    \text{s.t.}\;\;
    &\Pr_{(\mathcal{D},\mathcal{D}^{\rm v})\sim P_{\mathcal{D}\mathcal{D}^{\rm v}}}\!
    \left[A_{\mathcal{X}^{\rm v}}(\hat\phi(\mathcal{D}),\phi^{*})\ge
    A^{\rm min}_{\mathcal{D},\mathcal{D}^{\rm v}}\right]\ge 1-\epsilon,
    \label{decision}
\end{align}
where the optimization is over the learning rule $\hat\phi(\mathcal{D})$ and over the minimum guaranteed agreement $A^{\rm min}_{\mathcal{D},\mathcal{D}^{\rm v}}$.

The main result of this subsection is that problem \eqref{decision} is equivalent to an optimization over prediction sets of STL rules combined with a max-min decision rule. To elaborate, denote as $\mathcal{C}(\mathcal{D})$ a set of STL formulas obtained from dataset $\mathcal{D}$. Importantly, the set $\mathcal{C}(\mathcal{D})$ may include more than one STL formula, unlike output of standard STLL strategies (see Sec. \ref{sec:stll}). 

\noindent\underline{\textbf{Proposition 1:}}
The risk-averse problem \eqref{decision} is equivalent to the maximization  
\begin{align}
    \max_{\mathcal{C}(\cdot)}\;
    &\mathbb{E}_{(\mathcal{D},\mathcal{D}^{\rm v})\sim P_{\mathcal{D}\mathcal{D}^{\rm v}}}\!
    \left[
        \max_{\hat\phi\in\mathcal{C}(\mathcal{D})}\;\min_{\phi\in\mathcal{C}(\mathcal{D})}
        A_{\mathcal{X}^{\rm v}}(\hat\phi,\phi)
    \right]
    \notag\\
    \text{s.t.}\;\;
    & \Pr_{(\mathcal{D},\mathcal{D}^{\rm v})\sim P_{\mathcal{D}\mathcal{D}^{\rm v}}}\!
    \left[\phi^{*}\in\mathcal{C}(\mathcal{D})\right]
    \ge 1-\epsilon \label{ep}
\end{align}
over a set $\mathcal{C}(\mathcal{D})$ of STL rules.
\begin{proof}
    The proof follows directly from \cite[Theorem 2.3]{kiyani2025rac}.
\end{proof}

By Proposition 1, the optimal solution of problem \eqref{decision} requires the identification of a set $\mathcal{C}(\mathcal{D})$ of STL rules that includes the population-optimal rule $\phi^*$ with probability no smaller than $1-\epsilon$. Once such an optimal set $\mathcal{C}(\mathcal{D})$ is identified, the problem \eqref{ep} stipulates that the final STL rule $\hat{\phi}(\mathcal{D})$ be selected via the max-min rule
\begin{align}
    \hat{\phi}(\mathcal{D})=\arg\max_{\hat\phi\in\mathcal{C}(\mathcal{D})}\;\min_{\phi\in\mathcal{C}(\mathcal{D})}
        A_{\mathcal{X}^{\rm v}}(\hat\phi,\phi), \label{action}
\end{align}
i.e., as the STL rule that maximizes the worst-case validation agreement among all the STL rules in subset $\mathcal{C}(\mathcal{D})$. This theoretical insight is leveraged next to introduce a principled calibration strategy for STLL algorithms.   

\subsection{Conformal STL Learning} \label{sec:goal}
The risk-averse optimization problem \eqref{ep} is theoretically optimal in terms of validation agreement, but it cannot be solved due to its explicit dependence on the true data distribution. Furthermore, in practice, beside the agreement with an accuracy-maximizing STL rule, the designer should also target the interpretability of the rule.  In this subsection, we develop a practical solution that relies solely on access to calibration data \cite{angelopoulos2025learn, farzaneh2026statistically}, while also addressing the complexity of the learned STL rules.

To start, we relax the requirement \eqref{ep} by imposing that the set $\mathcal{C}(\mathcal{D})$ contains at least one STL rule $\hat{\phi}$ whose accuracy \eqref{accu} exceeds a user-defined level $\varphi$ with probability no smaller than $1-\epsilon$, i.e.,
\begin{align}
   \Pr\{\exists \hat \phi \in \mathcal{C}(\mathcal{D}): A_{\mathcal{D}^{\rm v}}(\hat\phi) \geq \varphi\} \geq 1-\epsilon. \label{relax}  
\end{align}
The coverage requirement \eqref{relax} relaxes the constraint \eqref{ep}, since the inequality \eqref{ep} implies \eqref{relax} whenever the threshold $\varphi$ is set below the accuracy of the population-optimal rule $\phi^*$. Unlike \eqref{ep}, which cannot be verified, since the population-optimal rule $\phi^*$ is unknown, the relaxed condition \eqref{relax} is verifiable from the validation labels.

The proposed C-STLL method constructs a set $\mathcal{C}(\mathcal{D})$ of candidate STL formulas satisfying \eqref{relax} by using multiple runs of a conventional STLL algorithm (Sec. \ref{sec:stll}) \cite{li2024tlinet}, thus implementing a form of test-time scaling \cite{muennighoff2025s1}.  As illustrated in Fig. \ref{model}, C-STLL calibrates an acceptance rule determining whether a formula generated via STLL is accepted for inclusion in the set, together with a stopping rule dictating when to stop generating. Specifically, as summarized in Algorithm \ref{rca}, given a dataset $\mathcal{D}$, C-STLL constructs a set of STL formulas $\mathcal{C}(\mathcal{D})$ via the following steps, which depend on a set of thresholds $\mv \lambda=(\lambda_1, \lambda_2, \lambda_3)$. To elaborate, denote as $\mathcal{C}_{\scalebox{0.7}{\mv \lambda},l}(\mathcal{D})$  the current set at iteration $l$.

\subsubsection{Sequential Generation and Acceptance/Rejection} \label{seq} At each learning run $l=1,2,...$, STLL is applied to the dataset $\mathcal{D}$ as described in Sec. \ref{stll}, obtaining a generally different formula $\phi_{l}$. Many STLL algorithms, such as evolutionary and genetic-programming approaches \cite{nenzi2018robust}, are inherently non-deterministic, and C-STLL exploits this property to obtain distinct candidate formulas $\phi_l$. 

At each run $l$, C-STLL accepts an STL formula $\phi_l$ into the set $\mathcal{C}_{\scalebox{0.7}{\mv \lambda},l}(\mathcal{D})$ if the formula is of sufficiently low complexity, and thus interpretable, and sufficiently diverse relative to the previously accepted formulas. To elaborate, let $H(\phi_{l})$ denote a function that maps the learned STL $\phi_{l}$ to a scalar quantity measuring its complexity, such as the number of operations in the formula \cite{bombara2016decision}. The candidate formula $\phi_{l}$ is deemed of sufficiently low complexity if the measure $H(\phi_{l})$ is smaller than a threshold  $\lambda_1$, i.e., H$(\phi_{l})<\lambda_1$. 

Furthermore, to avoid admitting redundant formulas, a new candidate $\phi_l$ is deemed to be sufficiently diverse from previous included formulas if the dissimilarity score $D(\phi_{l}, \phi)$ between $\phi_{l}$ and all of the previously accepted formulas $\phi\in\mathcal{C}_{\scalebox{0.7}{\mv \lambda}, l-1}(\mathcal{D})$ is larger than a threshold $\lambda_2$, i.e., $D(\phi_{l}, \phi) > \lambda_2, ~\text{for all}~\phi\in\mathcal{C}_{\scalebox{0.7}{\mv \lambda}, l-1}(\mathcal{D})$. 
Following \cite{madsen2018metrics}, we define a robustness-based dissimilarity score by comparing the robustness values of the two formulas on dataset $\mathcal{D}$ as
\begin{align}
    D(\phi_{l}, \phi)= \sigma \bigg(\frac{1}{|\mathcal{D}|} \sum_{i=1}^{|\mathcal{D}|} ||\rho(\mv X_i, \phi_l,0)|-|\rho(\mv X_i,\phi,0)| | \bigg), \label{distance}
\end{align}
with the sigmoid function $\sigma(\cdot)=(1+e^{-(\cdot)})^{-1}$.

When a candidate satisfies both the complexity criteria and diversity criteria, it is accepted into the set $\mathcal{C}_{\scalebox{0.7}{\mv \lambda},l}(\mathcal{D})$.
Accordingly, the set is updated as
\begin{equation}
\mathcal{C}^*_{\scalebox{0.7}{\mv \lambda},l}(\mathcal{D}) =
\begin{cases}
\mathcal{C}^*_{\scalebox{0.7}{\mv \lambda},l-1}(\mathcal{D}) \cup \{\phi_l\},
&
\begin{aligned}[t]
\text{if }&
H(\phi_l)<\lambda_1 \text{ and}\\
&\hspace{-2.9em}
D(\phi_l,\phi)>\lambda_2,\ 
\forall\,\phi\in
\mathcal{C}^*_{\scalebox{0.7}{\mv \lambda},l-1}(\mathcal{D}),
\end{aligned}
\\
\mathcal{C}^*_{\scalebox{0.7}{\mv \lambda},l-1}(\mathcal{D}),
&
\text{otherwise.}
\end{cases}
\end{equation}

The sequential generation process is run for at most $L_{\rm max}$ times, and may terminate early according to a robustness-based stopping score. We define a robustness-based stopping score for any STL formula $\phi_l$ as an increasing function of its average robustness on the dataset $\mathcal{D}$ as
\begin{align}
    Q(\phi_{l})=\sigma\bigg(\frac{1}{|\mathcal{D}|}\sum_{i=1}^{|\mathcal{D}|} \rho(\mv X_i, \phi_{l}, 0)\bigg), \label{qscore}
\end{align}
and define the set-based stopping score $F(\mathcal{C}_{\scalebox{0.7}{\mv \lambda},l}(\mathcal{D}))$ as the average 
\begin{align}
    F(\mathcal{C}_{\scalebox{0.7}{\mv \lambda},l}(\mathcal{D})) = \frac{1}{|\mathcal{C}_{\scalebox{0.7}{\mv \lambda},l}(\mathcal{D})|} \sum_{\phi \in\mathcal{C}_{\scalebox{0.6}{\mv \lambda},l}(\mathcal{D}) } Q(\phi).
\end{align}
C-STLL stops generating STL formulas at the earliest iteration $l$ for which we have the inequality
\begin{align}
    F(\mathcal{C}_{\scalebox{0.7}{\mv \lambda},l}(\mathcal{D})) > \lambda_3, \label{stop} 
\end{align}
where $\lambda_3$ is a threshold. We emphasize that $F(\cdot)$ is used only as a stopping score and is not intended as an estimator of classification accuracy. Its impact on reliability is accounted for through the joint LTT calibration of the thresholds.

The vector of hyperparameters $\mv \lambda=(\lambda_1, \lambda_2, \lambda_3)$, which encompasses the complexity threshold $\lambda_1$, the diversity threshold $\lambda_2$, and the early-stopping threshold $\lambda_3$ \eqref{stop}, is optimized during a preliminary calibration phase, which is discussed next.

\subsubsection{Reliability Requirements} \label{require}
C-STLL designs the hyperparameter $\mv \lambda$ of the acceptance and stopping rule with the aim of ensuring a reliability requirement \eqref{relax} for the set $\mathcal{C}_{\scalebox{0.7}{\mv \lambda}}(\mathcal{D})$, while also accounting for complexity and diversity of the formulas in set $\mathcal{C}_{\scalebox{0.7}{\mv \lambda}}(\mathcal{D})$.

To this end, we define the risk associated with the choice of hyperparameter $\mv \lambda$ as the probability that there is no sufficiently accurate STL formula $\hat{\phi}$ in set $\mathcal{C}_{\scalebox{0.7}{\mv \lambda}}(\mathcal{D})$, i.e.,
\begin{align}
    R(\mv \lambda)=  \Pr\{\nexists \hat{\phi} \in \mathcal{C}_{\scalebox{0.7}{\mv \lambda}}(\mathcal{D}): A_{\mathcal{D}^{\rm v}}(\hat\phi) \geq \varphi\}, \label{risk}
\end{align}
where $\varphi$ is the target accuracy in \eqref{relax} and the probability is taken over the unknown joint distribution $P_{\mathcal{D}\mathcal{D}^{\rm v}}$.  Our design goal is to identify a hyperparameter $\mv \lambda^*$ that guarantees that the requirement \eqref{relax} is controlled with high probability, i.e.,
\begin{align}
    \Pr(R(\mv \lambda^*) < \epsilon  ) \geq 1-\delta, \label{goal}
\end{align}
where the probability is over any randomness associated with the optimization of the hyperparameters $\mv \lambda^*$. By \eqref{goal}, the relaxed coverage \eqref{relax} holds at level $1-\epsilon$ with probability larger than $1-\delta$.

\begin{algorithm}[t]
  \caption{Conformal STL Learning (C-STLL)}\label{rca}
  \begin{algorithmic}[1]
    \STATE {\textbf{Input:} Labeled dataset $\mathcal{D}$} \\
    \STATE \textbf{Initialization:} The initial set is empty $\mathcal{C}_{\scalebox{0.7}{\mv \lambda},0}(\mathcal{D}) = \{\}$ \\
    \FOR{$l=1, 2,\ldots, L_{\rm max}$}
        \STATE Run STLL \cite{li2024tlinet} on dataset $\mathcal{D}$ to learn an STL formula $\phi_l$ \\
        \STATE \textbf{if}~ $H(\phi_l)<\lambda_1$~ \text{and}~ $D(\phi_{l}, \phi) > \lambda_2 ~\text{for all}~\phi\in\mathcal{C}_{\scalebox{0.7}{\mv \lambda}, l-1}(\mathcal{D})$
        \STATE ~~~~Add $\phi_{l}$ to the set $\mathcal{C}_{\scalebox{0.7}{\mv \lambda},l}(\mathcal{D})$
        \STATE \textbf{if}~ $F(\mathcal{C}_{\scalebox{0.7}{\mv \lambda},l}(\mathcal{D})) > \lambda_3$
        \STATE~~~~ Exit the for loop
        \STATE \textbf{end if}
    \ENDFOR
    \STATE {\textbf{Output:} Set $\mathcal{C}_{\scalebox{0.7}{\mv \lambda}}(\mathcal{D})=\mathcal{C}_{\scalebox{0.7}{\mv \lambda},l}(\mathcal{D})$ }
  \end{algorithmic}
\end{algorithm}

\subsection{Reliable Hyperparameter Optimization via Multiple Hypothesis Testing } \label{LTT}
As discussed in Sec. \ref{require}, C-STLL aims at identifying hyperparameters $\mv \lambda^*$ that satisfy the reliability requirement \eqref{goal}, while attempting to  introduce STL formulas with low complexity and to maintain diversity. To this end, C-STLL carries out an offline calibration procedure based on calibration data encompassing examples of dataset pairs $(\mathcal{D}, \mathcal{D}^{\rm v})$ from different tasks. Following the setting described in Section \ref{require}, we assume the availability of a calibration dataset $\mathcal{D}^{\rm cal}=\{(\mathcal{D}_k, \mathcal{D}_k^{\rm v})\}_{k=1}^{|\mathcal{D}^{\rm cal}|}$  consisting of pairs $(\mathcal{D}_k, \mathcal{D}_k^{\rm v})$ of training and validation datasets drawn i.i.d. from the dataset distribution $P_{\mathcal{D}\mathcal{D}^{\rm v}}$. Accordingly, each pair $(\mathcal{D}_k, \mathcal{D}_k^{\rm v})$ corresponds to an independently generated task $\tau_k \sim P_{\tau}$. C-STLL uses the calibration dataset $\mathcal{D}^{\rm cal}$ to select hyperparameter $\mv \lambda^*$ so that the condition \eqref{goal} is satisfied, where the outer probability is taken over the calibration dataset $\mathcal{D}^{\rm cal}$.

To this end, C-STLL follows the LTT methodology \cite{angelopoulos2025learn} by specifically adopting Pareto testing \cite{laufer2022efficiently}. In LTT, one starts by identifying a discrete subset $\Lambda$ of configurations for the hyperparameter $\mv \lambda$. This can be done via a standard grid or through any optimization algorithm using prior information or separate data. For each hyperparameter $\mv \lambda \in \Lambda$, LTT tests the null hypothesis $\mathcal{H}_{\scalebox{0.7}{\mv \lambda}}$ that the true miscoverage rate \eqref{risk} is above the target uncertainty level $\epsilon$, i.e., 
\begin{align}
    \mathcal{H}_{\scalebox{0.7}{\mv \lambda}}: R(\mv \lambda) \geq \epsilon. \label{hypo}
\end{align}
Rejecting the null hypothesis $\mathcal{H}_{\scalebox{0.7}{\mv \lambda}}$ identifies hyperparameter $\mv \lambda$ yielding reliable performance \eqref{relax}.

Using the calibration data $\mathcal{D}^{\rm cal}$, for each hyperparameter $\mv \lambda \in \Lambda$, we evaluate the empirical risk $\hat{R}(\mv \lambda, \mathcal{D}^{\rm cal})$ as
\begin{align}
    \hat{R}(\mv \lambda, \mathcal{D}^{\rm cal})= \frac{1}{|\mathcal{D}^{\rm cal}|}\sum_{k=1}^{|\mathcal{D}^{\rm cal}|}  \mathbbm{1}\{\nexists \hat\phi \in \mathcal{C}_{\scalebox{0.7}{\mv \lambda}}(\mathcal{D}_k): A_{\mathcal{D}^{\rm v}_k}(\hat\phi) \geq \varphi\}, \label{estimate}
\end{align}
which corresponds to an unbiased estimate of the true risk $R(\mv \lambda)$ in \eqref{risk}. Accordingly, if the null hypothesis $\mathcal{H}_{\scalebox{0.7}{\mv \lambda}}$ holds, the estimate \eqref{estimate} will tend to be larger than $\epsilon$. More precisely, let $b(|\mathcal{D}^{\rm cal}|, \epsilon)$ denote a binomial random variable with sample size $|\mathcal{D}^{\rm cal}|$ and success probability $\epsilon$. A valid p-value for the hypothesis $\mathcal{H}_{\scalebox{0.7}{\mv \lambda}}$ in \eqref{hypo} is given by \cite{quach2023conformal}
\begin{align}
    p(\mv \lambda)=\Pr(b(|\mathcal{D}^{\rm cal}|, \epsilon) \leq |\mathcal{D}^{\rm cal}|\hat{R}(\mv \lambda, \mathcal{D}^{\rm cal})), \label{pv}
\end{align}
in the sense that we have the inequality $\Pr(p(\mv \lambda) \leq \delta |\mathcal{H}_{\scalebox{0.7}{\mv \lambda}}) \leq \delta$ for all probability $\delta$.






Using the p-values $\{p(\mv \lambda)\}_{\scalebox{0.7}{\mv \lambda}\in \Lambda}$, LTT applies a family-wise error rate (FWER) multiple hypothesis testing (MHT) procedure to obtain a subset $\Lambda_{\rm valid} \subseteq \Lambda$ with the property 
\begin{align}
    \Pr\!\left( \exists\, \mv \lambda \in \Lambda_{\mathrm{valid}} :
R(\mv \lambda) \ge \epsilon \right) \le \delta. \label{FWER}
\end{align}
If the set $\Lambda^{\rm valid}$ is empty, no configuration in $\Lambda$ can be statistically certified at the prescribed level. In this case, we select the fallback solution $\mv \lambda^*=[\lambda_1=\infty, \lambda_2= 0, \lambda_3=\infty]^\top$, which allows the set $\mathcal{C}_{\scalebox{0.7}{\mv \lambda}}(\mathcal{D})$ to include all $L_{\max}$ STL formulas.  Otherwise, any configuration in $\Lambda^{\rm valid}$ guarantees the reliability condition \eqref{goal} by the FWER property \eqref{FWER} \cite{angelopoulos2025learn}. In C-STLL, we select the hyperparameter $\mv \lambda^* \in \Lambda_{\rm valid}$ that minimizes the empirical set size, i.e.,
\begin{align}
    \mv \lambda^*= \argmin_{\scalebox{0.7}{\mv \lambda}\in\Lambda_{\rm valid}} \frac{1}{|\mathcal{D}^{\rm cal}|}\sum_{k=1}^{|\mathcal{D}^{\rm cal}|} |\mathcal{C}_{\scalebox{0.7}{\mv \lambda}}(\mathcal{D}_k)|. \label{best}
\end{align}

As for the FWER procedure, in order to efficiently search the set of candidates $\Lambda$, we use the Pareto Testing procedure from \cite{laufer2022efficiently}. Pareto Testing exploits structure in $\Lambda$ by first using a proportion of the dataset $\mathcal{D}^{\rm cal}$ to approximate the Pareto-optimal frontier in two-dimensional space, and by then iteratively validating promising configurations
using fixed-trace testing \cite{wiens2003fixed} on the remaining calibration data.

\subsection{Theoretical Guarantees}
When $\Lambda^{\rm valid}\neq\emptyset$, C-STLL guarantees the target reliability condition \eqref{goal}.
\begin{theorem}[\textbf{Reliability of C-STLL via LTT}]
\label{theor}
If $\Lambda^{\rm valid}\neq\emptyset$, the hyperparameter vector
$\mv\lambda^\ast\in\Lambda^{\rm valid}$ selected according to \eqref{best} ensures that the risk of the set $\mathcal{C}_{\scalebox{0.7}{\mv\lambda}^\ast}(\mathcal{D})$
constructed via C-STLL satisfies \eqref{goal}.
\end{theorem}
\begin{proof}
    The proof is provided for completeness in the Appendix B.
\end{proof}

\section{Experiments} \label{sec:experiments}
In this section, we provide numerical results to demonstrate the effectiveness of the proposed C-STLL calibration methodology.

\subsection{Setting}
We evaluate C-STLL on two datasets. The first is generated with the network simulator ns-3 and is labeled by a ground-truth STL formula, making it possible to compare true and estimated formulas directly. The second consists of KPI traces 	measured on a softwarized RAN testbed (see Fig. \ref{fig:system-architecture}) and labeled by QoE scores reported by a commercial speed-test application \cite{ookla}.

\noindent \textbf{Synthetic dataset:} We adopt the network simulator ns-3 \cite{ns3} to model a single gaming user communicating with a remote host over a wireless network (see Fig. \ref{gaming}). Bidirectional user datagram protocol (UDP) traffic is generated to emulate interactive gaming, while optional background users introduce network congestion. The \emph{latency} KPI is defined as the one-way end-to-end delay of successfully received gaming packets, averaged over uplink and downlink. The \emph{backlog} KPI represents the number of user-side queued UDP packets awaiting transmission. Each trace $\mv X$ consists of latency and backlog KPIs collected over 61 time steps, which corresponds to uniformly spaced observations over a real time observation of 13 s.

A trace is labeled as positive $(Y=1)$ if the latency remains below a threshold $T_1$ over all time steps, while the backlog is below a threshold $T_2$ during the final 5 time steps; otherwise, it is labeled as negative $(Y=-1)$. This labeling approach corresponds to the ground-truth STL formula
\begin{align}
    \phi^{\rm true} &~= \Box_{[0,60]}(\text{latency}<T_1~ \text{ms}) \notag \\
    &~\land \Box_{[5 6,60]}(\text{backlog}<T_2~ \text{kilobyte}), \label{true}
\end{align}
We specifically use the threshold pairs  $(T_1, T_2) = (100, 30), (110, 28), (120, 26),  (130, 24)$, and $(140, 32)$  to label the entire dataset, corresponding to five distinct tasks. 

Each training dataset $\mathcal{D}$ contains 5,000 examples, and each validation set $\mathcal{D}^{\rm v}$ contains 1,000 examples. Each dataset pair $(\mathcal{D}, \mathcal{D}^{\rm v})$ is generated by first uniformly sampling a task $\tau$, and then sampling from the labeled data corresponding to this task. We generate 120 calibration data pairs $(\mathcal{D}, \mathcal{D}^{\rm v})$, of which 50 pairs are used for Pareto testing, and the remaining 70 for fixed-sequence testing (see Section \ref{LTT}). All test metrics are averaged over 200 independently sampled tasks.

\begin{figure}[ht]
    \centering
    \includegraphics[width=1.0\linewidth]{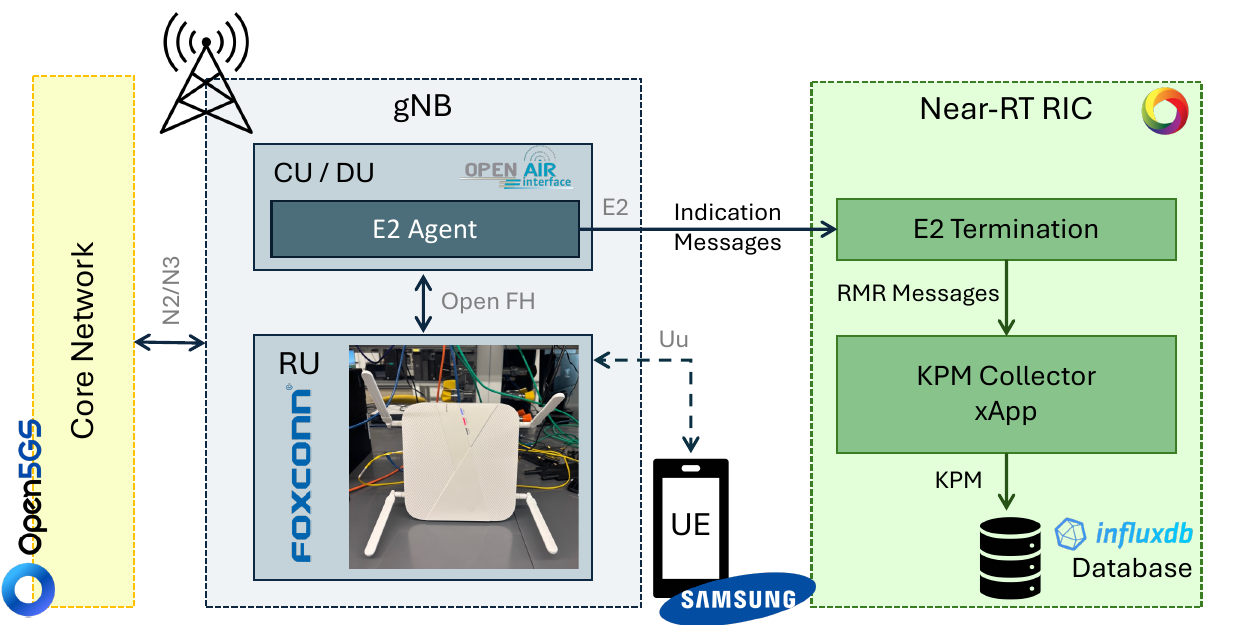}
    \caption{System architecture of the Over-The-Air testbed. The OAI gNB, split into O-CU/O-DU and a Foxconn O-RU, serves a commercial UE and connects to an Open5GS core network. The FlexRIC E2 agent reports E2SM-KPM measurements over E2 to the OSC Near-RT RIC, where the KPM Collector xApp receives them via RMR and stores them in InfluxDB. All components are deployed as containerized workloads on an OpenShift cluster.}
    \label{fig:system-architecture}
\end{figure}

\noindent \noindent \textbf{Real-world dataset:} The second dataset was collected on an indoor 5G testbed~\cite{Villa_2025}, whose end-to-end architecture is illustrated in Fig.~\ref{fig:system-architecture}. The testbed follows the O-RAN disaggregated architecture, and is deployed entirely as containerized workloads on an OpenShift cluster~\cite{maxenti2025autoranautomatedzerotouchopen}. The RAN is realized with OpenAirInterface (OAI) 5G protocol stack following the O-RAN 7.2x functional split. Accordingly, the Central and Distributed Units (O-CU/O-DU) runs in a dedicated deployment, while the lower physical layer functions are executed by a Foxconn Radio Unit (RU) connected over the Open Fronthaul (FH) interface. A commercial Samsung xcover 7 pro phone attaches to the cell over the Uu interface on band [n78] with 40 MHz of bandwidth. The 5G core network is provided by Open5GS, deployed on the same cluster and connected to the O-CU/O-DU through the standard N2 and N3 interfaces for control- and user-plane traffic respectively.

In this setup, a monitoring application records the downlink throughput and delay, while the user concurrently runs a commercial speed-test application \cite{ookla} that rates the measured connection for video streaming on a five-level scale. Each trace $\mv X$ consists of the two KPIs, namely throughput and delay, over 50 time steps, each corresponding to 1 s. A trace is labeled positive ($Y = 1$) if its QoE score is at least $k$, and negative ($Y = -1$) otherwise, for $k = 2, 3, 4, 5$, yielding four tasks of 5,000 traces each. Dataset pairs $(\mathcal{D}, \mathcal{D}^{\rm v})$ are generated as for the synthetic example, with 1,500 and 800 samples, respectively.

\begin{table*}[t!] 
\centering
\caption{Examples of learned STL formula sets with conventional STLL \cite{li2024tlinet} and with variants of the proposed C-STLL calibration schemes given the same training dataset $\mathcal{D}$ (true labeling STL formula \eqref{true} with $T_1=100$ and $T_2=30$). }
\label{tab:learned_stl_sets}
\renewcommand{\arraystretch}{1.2}
\setlength{\tabcolsep}{8pt}

\begin{tabularx}{\textwidth}{p{0.12\textwidth}|X|p{0.09\textwidth}}
\hline
\textbf{Scheme} &
\textbf{Learned STL Formula Set $\mathcal{C}_{\scalebox{0.7}{\mv \lambda}^*}(\mathcal{D})$} &
\textbf{Accuracy} \\
\hline

\textbf{STLL} &
\makecell[tl]{%
\(\phi=\ 
\begin{aligned}[t]
& \Box_{[24,39]}(\text{latency}<21)\ \wedge\ 
  \Box_{[1,57]}(\text{backlog}<59)\ \wedge\ 
  \Box_{[29,30]}(\text{latency}<66)\ \wedge\\
& \Box_{[11,42]}(\text{backlog}>3)
\end{aligned}
\)
} &
\makecell[tl]{
\(\phi: 65\%\)
} \\
\hline

\textbf{C-STLL} &
\makecell[tl]{
\(\{\)\(\phi_1=\ \Box_{[20,31]}(\text{latency}<137)\ \wedge\ \Box_{[25,60]}(\text{backlog}<55)\),\\
\(\phi_2=\ \Diamond_{[11,40]}(\text{backlog}>3)\ \wedge\ \Box_{[0,55]}(\text{latency}<30)\ \)\(\}\)
} &
\makecell[tl]{
\(\phi_1: \textbf{96.9}\%\),\\
\(\phi_2: \textbf{85\%} \)
} \\
\hline

\textbf{Basic C-STLL} &
\makecell[tl]{%
\(\{\)\(\phi_1=\ \Box_{[19,49]}(\text{backlog}<64)\),\\
\(\phi_2=\ \Diamond_{[19,33]}(\text{latency}<154)\ \vee\ \Box_{[9,48]}(\text{latency}<22)\),\\
\(\phi_3=\ \Diamond_{[13,35]}(\text{backlog}>60)\ \vee\ \Diamond_{[14,45]}(\text{latency}<20)\),\\
\(\phi_4=\ \Diamond_{[11,60]}(\text{backlog}<4)\ \vee\ \Box_{[13,38]}(\text{backlog}<63)\),\\
\(\phi_5=\ \Box_{[5,29]}(\text{backlog}<81)\),\\
\(\phi_6=\ \Diamond_{[18,32]}(\text{backlog}>16)\ \wedge\ \Diamond_{[24,54]}(\text{backlog}>11)\ \wedge\ \Diamond_{[0,33]}(\text{backlog}<45)\),\\
\(\phi_7=\ \Box_{[15,36]}(\text{latency}<14)\ \wedge\ \Diamond_{[17,53]}(\text{backlog}<24)\),\\
\(\phi_8=\ \Box_{[25,29]}(\text{latency}<27)\),\\
\(\phi_9=\ \Box_{[0,33]}(\text{backlog}>11)\ \vee\ \Box_{[8,46]}(\text{backlog}>21)\),\\
\(\phi_{10}=\ \Box_{[15,38]}(\text{backlog}>2)\ \wedge\ \Box_{[18,29]}(\text{latency}<28)\)\(\}\)
} &
\makecell[tl]{%
\(\phi_1: \textbf{84.1\%}\),\\
\(\phi_2: \textbf{84.1\%}\),\\
\(\phi_3: 52.7\%\),\\
\(\phi_4: 57.5\%\),\\
\(\phi_5: \textbf{84.1\%}\),\\
\(\phi_6: \textbf{81.3\%}\),\\
\(\phi_7: 47.3\%\),\\
\(\phi_8: 76.8\%\),\\
\(\phi_9: 52.7\%\),\\
\(\phi_{10}: 71.9\%\)
} \\
\hline

\textbf{C-STLL with complexity control only} &
\makecell[tl]{%
\(\{\)\(\phi_1=\ \Diamond_{[23,45]}(\text{backlog}>2)\ \wedge\ \Box_{[21,39]}(\text{backlog}<75)\),\\
\(\phi_2=\ \Diamond_{[24,39]}(\text{backlog}>8)\ \wedge\ \Box_{[27,46]}(\text{latency}<32)\),\\
\(\phi_3=\ \Diamond_{[18,32]}(\text{latency}<17)\ \wedge\ \Diamond_{[26,52]}(\text{backlog}>4)\),\\
\(\phi_4=\ \Diamond_{[27,33]}(\text{backlog}<45)\),\\
\(\phi_5=\ \Diamond_{[17,47]}(\text{backlog}>12)\ \wedge\ \Diamond_{[1,35]}(\text{backlog}<73)\),\\
\(\phi_6=\ \Diamond_{[1,44]}(\text{latency}<17)\ \wedge\ \Diamond_{[13,52]}(\text{backlog}>13)\),\\
\(\phi_7=\ \Box_{[25,54]}(\text{latency}<20)\ \wedge\ \Diamond_{[5,36]}(\text{backlog}>6)\)\(\}\)
} &
\makecell[tl]{%
\(\phi_1: \textbf{84.1\%} \),\\
\(\phi_2: 78.6\%\),\\
\(\phi_3:53.8\%\),\\
\(\phi_4: \textbf{82.6\%}\),\\
\(\phi_5:52.7\%\),\\
\(\phi_6: 57.7\%\),\\
\(\phi_7:48.1\%\)
} \\
\hline

\textbf{C-STLL with diversity control only} &
\makecell[tl]{%
\(\{\)\(\phi_1=\ 
\begin{aligned}[t]
& \Box_{[24,51]}(\text{latency}<32)\ \wedge\ 
  \Diamond_{[9,46]}(\text{latency}>7)\ \wedge\ 
  \Box_{[13,37]}(\text{backlog}>369)\ \wedge\\
& \Box_{[4,36]}(\text{latency}<31)
\end{aligned}
\)\\
\(\phi_2=\ \Box_{[1,53]}(\text{latency}<29)\ \vee\ \Diamond_{[0,58]}(\text{latency}<28)\),\\
\(\phi_3=\ 
\begin{aligned}[t]
& \Diamond_{[11,45]}(\text{latency}<179)\ \wedge\ 
  \Box_{[19,55]}(\text{backlog}>48)\ \wedge\ 
  \Box_{[23,56]}(\text{latency}>10)\ \wedge\\
& \Diamond_{[25,54]}(\text{backlog}>195)\ \wedge\ 
  \Diamond_{[0,52]}(\text{latency}<101)\ \wedge\ 
  \Box_{[9,32]}(\text{backlog}>189) \}
\end{aligned}
\)
} &
\makecell[tl]{%
\(\phi_1: 70.6\% \),\\
\(\phi_2: \textbf{83.9\%}\),\\
\(\phi_3: \textbf{83.9\%}\)
} \\
\hline

\end{tabularx}
\end{table*}

\subsection{Implementation}
We adopt the STLL method in \cite{li2024tlinet}, which learns an STL formula via neural network–based optimization, as explained in Section \ref{STLL} and Appendix A. We follow a template similar to Fig. \ref{TLI}, with six predicate modules at the input layer, followed by six temporal modules and a single Boolean module as the output layer.  The learning rate is set to 0.1, the batch size to 64, and the number of training epochs to 100, and in  the loss function \eqref{loss}, we set $\beta=0.1$ and $\gamma=0.01$. The grid $\Lambda$ contains all threshold tuples $(\lambda_1, \lambda_2, \lambda_3)$ with $\lambda_1 \in \{0.33, 0.5, 0.67\}$, $\lambda_2 \in \{0.50, 0.52, 0.54, 0.56, 0.58\}$, and $\lambda_3\in\{0.3, 0.4, 0.5, 0.6, 0.7\}$. These ranges were chosen based on preliminary experiments to identify regions with meaningful variation in risk and set size over separate data generated in the same way. The complexity measure $H(\phi)$ used for calibration is normalized to $[0,1]$, whereas the figures report the corresponding unnormalized number of operations for interpretability.

For C-STLL, unless otherwise stated, the accuracy threshold in \eqref{relax} is set to $\varphi=0.8$, thus requiring that a fraction larger than 80\% of validation data points is correctly classified. The error level in \eqref{goal} is set to $\delta=0.05$, so that the probability of selecting an unreliable hyperparameter $\mv \lambda^*$ is no more than 0.05, and C-STLL is run for at most $L_{\rm max}=10$ iterations.

\subsection{Benchmarks}
We compare the following schemes:

\noindent $\bullet$ \textbf{STLL:} This baseline learns a single STL formula by minimizing a robustness-based loss \eqref{loss} over the training dataset $\mathcal{D}$ following \cite{li2024tlinet}. Unlike C-STLL, STLL provides no explicit reliability control, and is obtained from Algorithm \ref{rca} by setting $L_{\rm max}=1$.

\noindent $\bullet$  \textbf{Basic C-STLL:} This variant of the proposed C-STLL scheme uses the stopping rule based on the set-based stopping score $F$, but it does not apply the complexity and diversity control when accepting individual formulas $\phi_l$. Accordingly, in this case, only the stopping threshold $\lambda_3$ is optimized, while thresholds $\lambda_1$ and $\lambda_2$ are not used in Algorithm \ref{rca}.

\noindent $\bullet$ \textbf{C-STLL with complexity control only:} This variant of C-STLL enforces the complexity-based acceptance rule, admitting a candidate formula $\phi_l$ only if its complexity is smaller than the threshold $\lambda_1$, but no diversity constraint is applied. In this case, only the complexity threshold $\lambda_1$ and the stopping threshold $\lambda_3$ are optimized, while threshold $\lambda_2$ is not used in Algorithm \ref{rca}.

\noindent $\bullet$ \textbf{C-STLL with diversity control only:} This variant enforces the diversity-based acceptance rule, requiring each accepted formula $\phi_l$ to be sufficiently dissimilar from previously accepted formulas, but no complexity constraint is applied. In this case, only the diversity threshold $\lambda_2$ and the stopping threshold $\lambda_3$ are optimized, while threshold $\lambda_1$ is not used in Algorithm \ref{rca}.

We emphasize that all C-STLL variants are novel, and that prior art only considered STLL \cite{li2024tlinet}.  The benchmarks above differ in how the set $\mathcal{C}_{\scalebox{0.7}{\mv \lambda}}(\mathcal{D})$ is built. For C-STLL, we also compare the following alternatives for extracting a decision from the obtained set of STL rules.

\noindent $\bullet$ \textbf{Max-min rule:} This is the default proposed selection using the derived max-min rule \eqref{action}.

\noindent $\bullet$ \textbf{Ensembling:} Instead of deploying a single STL rule, the decisions of all formulas in $\mathcal{C}_{\scalebox{0.7}{\mv \lambda}}(\mathcal{D})$ are combined. Specifically, under majority-voting ensembling, a trace $\mv{X}$ is assigned the label $1$ if the condition $\sum_{\phi} \mathbbm{1}(c_{\phi}(\mv{X}) = 1) \geq |\mathcal{C}_{\scalebox{0.7}{\mv \lambda}}(\mathcal{D})|/2$ holds.

\noindent $\bullet$ \textbf{Random selection: } One formula is drawn uniformly from
$\mathcal{C}_{\scalebox{0.7}{\mv \lambda}}(\mathcal{D})$ per dataset pair, and used to classify all traces.

\subsection{Results on Synthetic Data} \label{res}
We start by considering the ns-3 setting, assuming that data are labeled via the ground-truth STL formula $\phi^{\rm true}$ in \eqref{true}.  To illustrate the outputs of different schemes, in Table I, we present examples of STL formula sets learned using the same dataset $\mathcal{D}$ under STLL, which returns a single formula, as well as with the mentioned C-STLL variants. As discussed, the C-STLL variants implement different acceptance and stopping mechanisms, which directly shape the resulting set. Table I also reports the accuracy on the same validation dataset $\mathcal{D}^{\rm v}$ as in \eqref{accu}, i.e., $(\sum_{i=1}^{|\mathcal{D}^{\rm v}|}\mathbbm{1}(\hat{Y}_i=Y_i))/|\mathcal{D}^{\rm v}|$, marking in bold the STL formulas with accuracy above the target threshold $\varphi=0.8$.

\begin{figure*}[t!]
    \centering
    \includegraphics[width=0.9\linewidth]{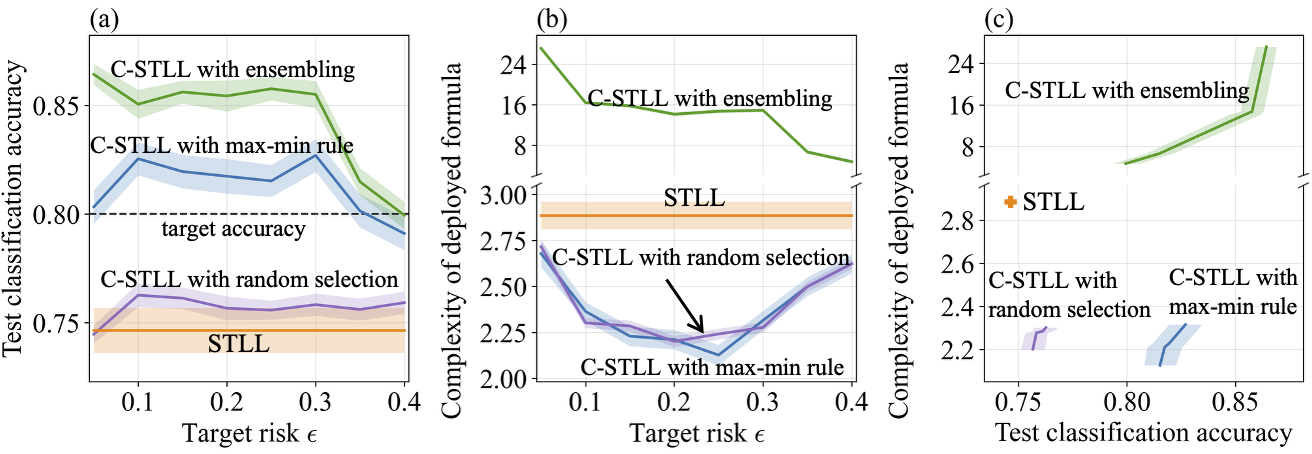}
    \caption{Synthetic dataset results based on the ns-3 simulator for STLL \cite{li2024tlinet} and C-STLL with different decision rules:  (a) test classification accuracy; (b) complexity of the deployed rules; and (c) trade-off between test classification accuracy and deployment complexity for target accuracy $\varphi=0.8$. }
    \label{new1}
\end{figure*}

\begin{figure*}[t!]
    \centering
    \includegraphics[width=0.9\linewidth]{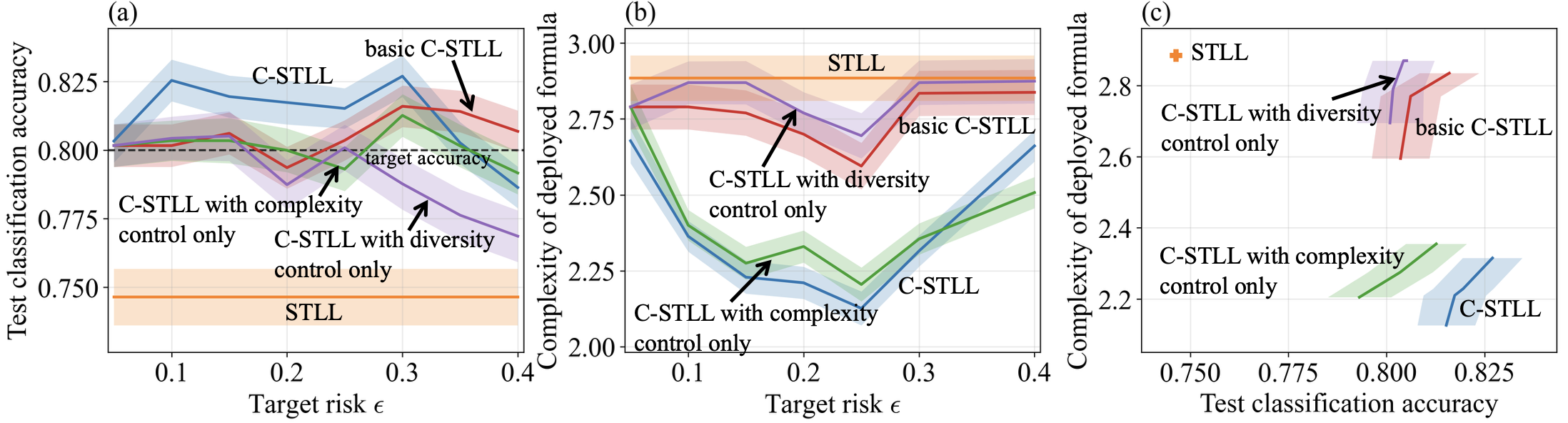}
    \caption{Synthetic dataset results based on the ns-3 simulator for STLL \cite{li2024tlinet} and C-STLL variants:  (a) test classification accuracy; (b)  complexity of the deployed rules; and (c) trade-off between test classification accuracy and deployment complexity for target accuracy $\varphi=0.8$. }
    \label{new2}
\end{figure*}

\begin{figure}[t!]
    \centering
    \includegraphics[width=1\linewidth]{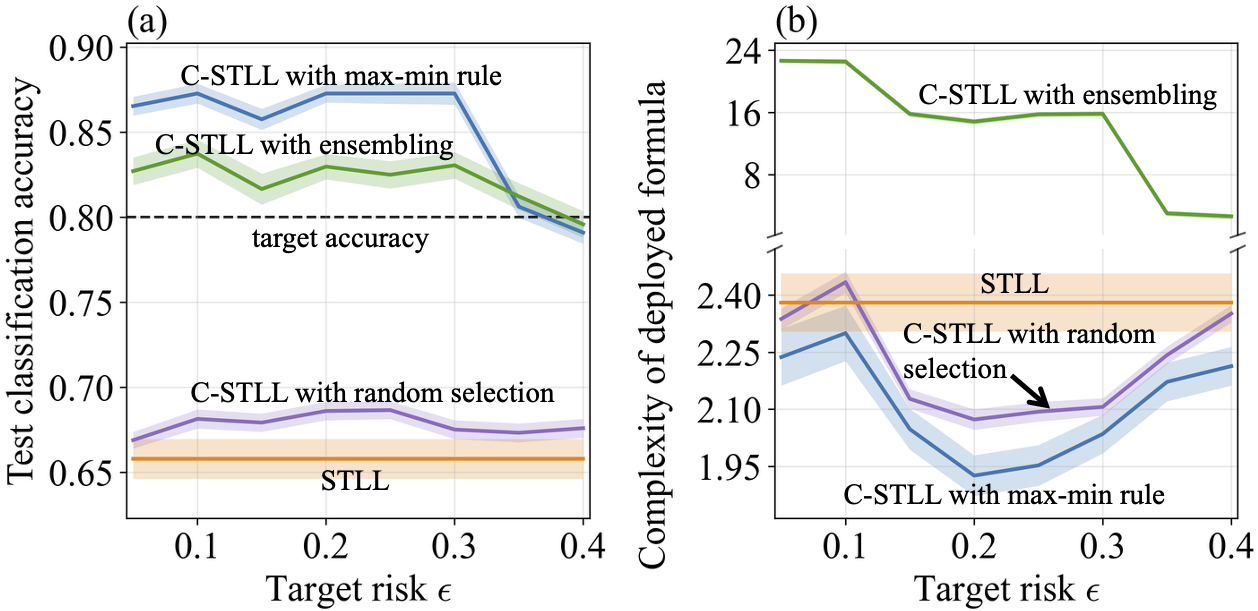}
    \caption{Real-world dataset results based on testbed for STLL \cite{li2024tlinet} and C-STLL with different decision rules:  (a) test classification accuracy; and (b) complexity of the deployed rules for target accuracy $\varphi=0.8$. }
    \label{new3}
\end{figure}

\begin{figure}[t!]
    \centering
    \includegraphics[width=1.0\linewidth]{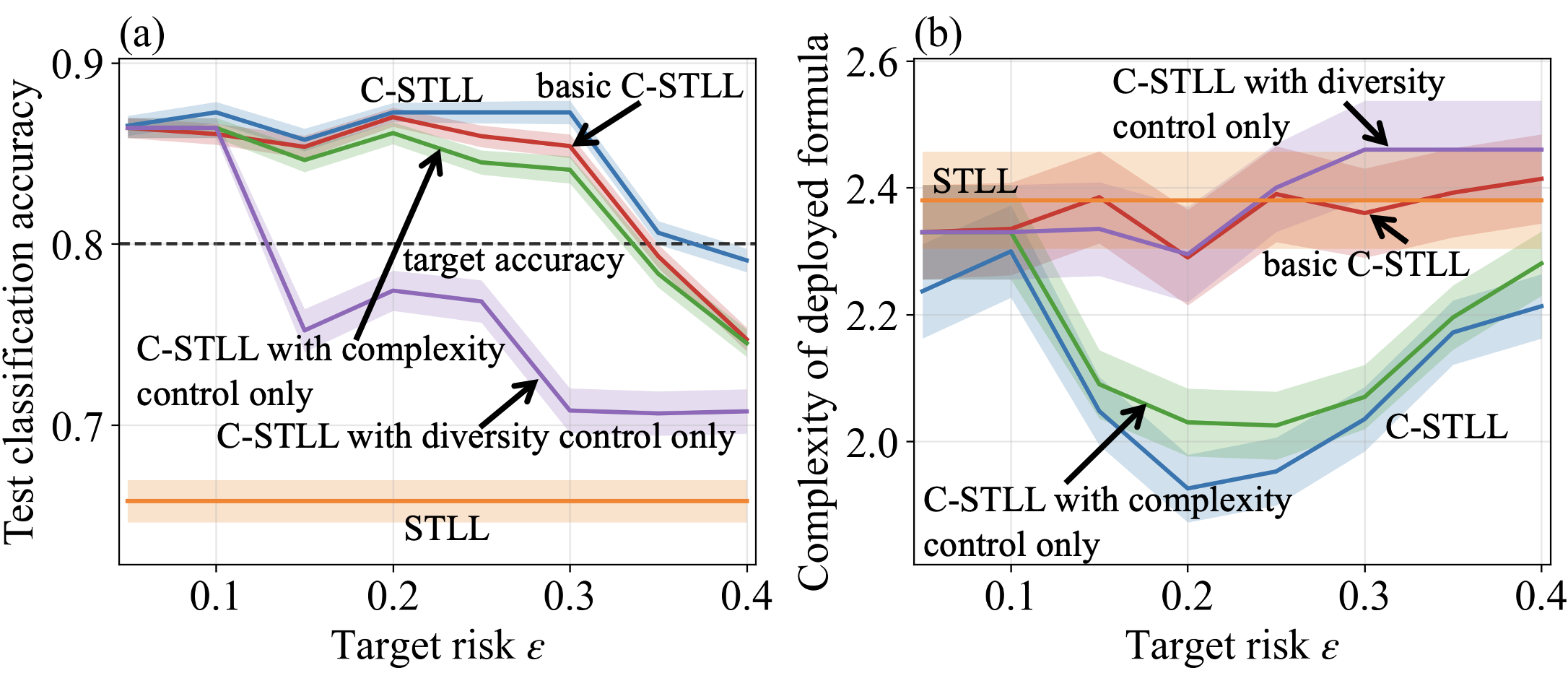}
    \caption{Real-world dataset results based on testbed for STLL \cite{li2024tlinet} and C-STLL variants:  (a) test classification accuracy; and (b) complexity of the deployed rules for target accuracy $\varphi=0.8$. }
    \label{new4}
\end{figure}

STLL returns a single formula, which is seen to be quite different from the ground-truth STL formula $\phi^{\rm true}$ in \eqref{true}, yielding an accuracy level below the target 80\%.  In contrast, the full C-STLL scheme results in a compact set, where each formula is both low in complexity and complementary to the others, resulting in at least one formula in the set with accuracy  above the threshold $\varphi$.

When only the stopping rule is used, C-STLL admits all learned formulas until the stopping criterion is met. As a result, the corresponding set is large, and contains many formulas that are structurally similar or redundant. Introducing a complexity control check restricts the admission of formulas with many temporal or Boolean operators, leading to a smaller set composed of simpler and more interpretable formulas. In contrast, enforcing a diversity check suppresses formulas whose robustness behaviors are similar to previously accepted ones, encouraging structural variety but still allowing relatively complex formulas to enter the set. This analysis validates the importance of including complexity and diversity checks.

We next investigate how to leverage the set of formulas produced by C-STLL for test-time classification. To this end, Fig.~\ref{new1} plots: (a) the classification accuracy and (b) the average complexity of the deployed formulas in terms of number of operations versus the target risk $\epsilon$, as well as (c) the resulting trade-off between test accuracy and complexity. As shown in Fig.~\ref{new1}(a), for sufficiently small values of the target risk $\epsilon$, C-STLL schemes attain the required accuracy level of $\varphi=0.8$ (dashed line), while higher target risk levels $\epsilon$ cause test accuracy to drop below $\varphi$. In contrast, STLL fails to guarantee any accuracy requirement. On the flip side, as seen in Fig.~\ref{new1}(b), larger values of $\epsilon$ can yield lower-complexity formulas by reducing the STL set size. In this regard, C-STLL with ensembling requires a significantly larger complexity, involving all STL formulas within the set, whereas the proposed max-min rule produces much simpler formulas, e.g., with $2.2$ operations at $\epsilon=0.2$, corresponding to accuracy 82\%. 

The overall trade-off between test accuracy and complexity is shown in Fig.~\ref{new1}(c) by varying the target risk $\epsilon$. C-STLL with the max-min rule is seen to provide the most favorable compromise between the two criteria yielding a higher accuracy at lower complexity levels. Conventional STLL, instead, provides a single operational point characterized by high complexity and low accuracy.

We next examine the contributions of the complexity and diversity controls to the performance of C-STLL with the max-min rule in Fig.~\ref{new2}. From Fig.~\ref{new2}(a), simultaneous complexity and diversity controls are seen to be instrumental in maintaining the test accuracy level $\varphi=0.8$ across a range of values for the target risk $\epsilon$. Individual complexity or diversity checks can cause a faster decrease of the accuracy as a function of $\epsilon$.

A larger $\epsilon$ is seen in Fig.~\ref{new2}(b) to be initially beneficial in reducing the complexity of the deployed formulas. The complexity control check is potentially effective in reducing complexity, although the lack of diversity control yields a larger set and a lower test accuracy. The overall trade-off between test accuracy and complexity of the deployed formula shown in Fig.~\ref{new2}(c)  confirms the effectiveness of the simultaneous application of complexity and diversity checks in producing the most effective and interpretable, i.e., less complex, STL rules.

\subsection{Results on Real-World Data}

We now evaluate the different decision strategies on the real-world testbed dataset. The results, reported in Fig.~\ref{new3}, align with the conclusions derived above for the synthetic dataset. To highlight some numerical exemplifying outcomes, as shown in Fig.~\ref{new3}(a), at target risk $\epsilon=0.2$, C-STLL with the max-min rule achieves an accuracy of around $87\%$, well above the target $80\%$, compared with around $83\%$ for C-STLL with ensembling, $68\%$ for C-STLL with random selection, and $66\%$ for STLL. Moreover, as shown in Fig.~\ref{new3}(b), the complexity of the formula selected by C-STLL with the max-min rule is only around $1.9$, whereas ensembling requires a complexity of around $15$ operations. 

Finally, Fig.~\ref{new4} investigates the contribution of the acceptance criteria on the real-world dataset. Focusing again on exemplifying numerical results, at target risk $\epsilon=0.2$, the complete C-STLL scheme achieves an accuracy of around $86\%$ with a deployed-formula complexity of around $2$ operations, whereas conventional STLL achieves only around $66\%$ accuracy with a complexity of around $2.4$ operations. Basic C-STLL achieves a similar accuracy to the complete C-STLL, but with a higher complexity of around $2.3$ operations. In contrast, using complexity or diversity control alone results in a marked accuracy degradation. This confirms that the two criteria are complementary:  complexity control favors simpler formulas but does not prevent redundancy, while diversity control promotes distinct formulas without explicitly limiting their complexity. Their joint use therefore produces a more informative candidate set for the subsequent max-min selection, resulting in more advantageous trade-offs between accuracy and interpretability.

\section{Conclusions} \label{sec:conclusion}

This paper has introduced conformal signal temporal logic learning (C-STLL), a novel framework for extracting interpretable temporal specifications from wireless network KPI traces with formal reliability guarantees. By combining the expressiveness of STL with the distribution-free coverage guarantees of modern hyperparameter selection methods, C-STLL addresses both \emph{relevance} -- capturing temporal patterns predictive of high-level requirements such as QoE -- and \emph{interpretability} -- providing human-readable specifications that network operators can validate, understand, and act upon.

Experimental results on a mobile gaming scenario simulated using ns-3 and on KPI traces measured on a softwarized RAN testbed have demonstrated the effectiveness of C-STLL.  Compared to standard STLL \cite{li2024tlinet}, which returns a single formula without reliability assurances, C-STLL produced compact sets of diverse, low-complexity formulas that achieved high accuracy on validation data. Ablation studies confirmed the importance of each component: complexity checks yielded more interpretable formulas, diversity checks prevented redundancy, and the stopping rule enabled efficient early termination once the stopping criterion was met.

Several directions merit further investigation. First, extending C-STLL to an \emph{online} setting, where KPI traces arrive sequentially and the formula set must be updated incrementally, would enhance its applicability to real-time network monitoring. Second, investigating the use of large language models to generate natural-language explanations of learned STL formulas could further bridge the gap between formal specifications and operator understanding.

\section*{Appendix A: Temporal Logic Inference Network}

This appendix reviews TLINet \cite{li2024tlinet}, which learns STL formulas by integrating quantitative STL semantics with neural network optimization.

\noindent \textbf{vSTL Representation: } TLINet uses vectorized STL (vSTL), where Boolean operators $\land / \lor$ and temporal operators $\Box /\Diamond$ are parameterized by binary vectors. A binary vector $\mv w_b \in \{0,1\}^n$ selects active subformulas in Boolean operations, while $\mv w_I \in \{0,1\}^T$ encodes the time interval $[t_1, t_2]$ for temporal operators, with $w^t_I = 1$ if $t_1 \leq t \leq t_2$. The robustness computations in \eqref{s1}–\eqref{te1} are reformulated accordingly over these binary vectors.

\noindent \textbf{TLINet Architecture:} As illustrated in Fig. \ref{TLI}, TLINet contains predicate, temporal, and Boolean modules arranged in layers. The predicate module takes trace $\mv X$ as input and outputs robustness vectors using trainable parameters $\mv a$ and $b$. Temporal modules learn the interval bounds $t_1, t_2$ via smooth approximations of $\mv w_I$, and select between $\Box$ and $\Diamond$ using a Bernoulli variable $\kappa$ with trainable probability $p_{\kappa}$. Boolean modules similarly learn subformula selection via Bernoulli variables for $\mv w_b$ and operator type, either $\land$ or $\lor$, via another Bernoulli variable $\kappa_b$. The output robustness $\rho(\mv X, \phi, 0)$ serves as the decision score. The training objective combines the hinge classification loss in \eqref{loss} with regularizers promoting binarization of all probabilities via $p(1-p)$ terms and sparsity in Boolean selection via $\ell_1$ regularization on $p_b$. Parameters are optimized using stochastic gradient descent.

\section*{Appendix B: Proof of Theorem 1}

The reliability condition \eqref{goal} follows from the validity of the p-value in \eqref{pv} combined with the FWER control of fixed-trace testing.

To show that the p-value is valid, let $S=|\mathcal{D}^{\rm cal}|\hat{R}(\mv \lambda, \mathcal{D}^{\rm cal})$ and note that $S \sim b(|\mathcal{D}^{\rm cal}|, R(\mv \lambda))$ since calibration pairs are i.i.d. For any $\alpha \in [0, 1]$, because the binomial CDF $F_{\epsilon}(\cdot)$ is nondecreasing, the event $\{F_{\epsilon}(S)\leq \alpha\}$ implies $S\leq c_{\alpha}$, where $c_{\alpha}$ is the $\alpha$-quantile index. Under the null hypothesis \eqref{hypo}, the binomial CDF is nonincreasing in its success probability, so $\Pr(S \leq c_{\alpha}) \leq F_{\epsilon}(c_{\alpha}) \leq \alpha$. This confirms the inequality $\Pr(p(\mv \lambda) \leq \alpha | \mathcal{H}_{\scalebox{0.7}{\mv \lambda}}) \leq \alpha$.

Since fixed-trace testing controls the FWER at level $\delta$ \cite{wiens2003fixed}, when $\Lambda^{\rm valid}\neq\emptyset$, the selected $\mv\lambda^\ast$ satisfies the reliability condition \eqref{goal} with probability at least $1-\delta$.

\small{
\bibliographystyle{ieeetr}
\bibliography{references}
}

\end{document}